\documentclass[sigconf,authorversion,nonacm]{acmart}

\usepackage{graphicx,mathtools,kantlipsum}
\usepackage[normalem]{ulem}
\usepackage{url}
\usepackage{amsmath,amsfonts}
\usepackage{listings}
\usepackage{color}
\usepackage{courier}
\usepackage{float}
\usepackage{makecell}
\usepackage{algorithm}
\usepackage[noend]{algpseudocode}
\usepackage{todonotes}
\usepackage{booktabs}
\usepackage{caption}
\usepackage{subcaption}
\usepackage{hyperref}
\usepackage[capitalize]{cleveref}
\usepackage{balance}
\usepackage{enumitem}
\usepackage{bm}
\usepackage{multicol}
\usepackage{multirow}
\usepackage{xspace}
\usepackage{tabularx}
\usepackage{colortbl}
\AtBeginDocument{%
    \providecommand\BibTeX{{%
        \normalfont B\kern-0.5em{\scshape i\kern-0.25em b}\kern-0.8em\TeX}}}

\usepackage{tikz}
\usepackage{pgfplots}
\usepackage{pgfplotstable}
\pgfplotsset{compat=newest}
\pgfplotsset{plot coordinates/math parser=false}
\pgfplotsset{compat = 1.3}

\usepgfplotslibrary{colorbrewer}

\usetikzlibrary{pgfplots.statistics, pgfplots.colorbrewer}
\usetikzlibrary{patterns}
\usetikzlibrary{tikzmark}
\usetikzlibrary{arrows}
\usetikzlibrary{pgfplots.groupplots}

\tikzstyle{every picture}+=[remember picture]
\everymath{\displaystyle}

\makeatletter
\pgfplotsset{
    groupplot xlabel/.initial={},
    every groupplot x label/.style={
        at={($({group c1r\pgfplots@group@rows.west}|-{group c1r\pgfplots@group@rows.outer south})!0.5!({group c\pgfplots@group@columns r\pgfplots@group@rows.east}|-{group c\pgfplots@group@columns r\pgfplots@group@rows.outer south})$)},
        anchor=north,
    },
    groupplot ylabel/.initial={},
    every groupplot y label/.style={
            rotate=90,
        at={($({group c1r1.north}-|{group c1r1.outer
west})!0.5!({group c1r\pgfplots@group@rows.south}-|{group c1r\pgfplots@group@rows.outer west})$)},
        anchor=south
    },
    execute at end groupplot/.code={%
      \node [/pgfplots/every groupplot x label]
{\pgfkeysvalueof{/pgfplots/groupplot xlabel}};
      \node [/pgfplots/every groupplot y label]
{\pgfkeysvalueof{/pgfplots/groupplot ylabel}};
    },
    group/only outer labels/.style =
{
group/every plot/.code = {%
    \ifnum\pgfplots@group@current@row=\pgfplots@group@rows\else%
        \pgfkeys{xticklabels = {}, xlabel = {}}\fi%
    \ifnum\pgfplots@group@current@column=1\else%
        \pgfkeys{yticklabels = {}, ylabel = {}}\fi%
}
}
}

\def\endpgfplots@environment@groupplot{%
    \endpgfplots@environment@opt%
    \pgfkeys{/pgfplots/execute at end groupplot}%
    \endgroup%
}
\makeatother

\pgfplotsset{every tick label/.append style={font=\normalsize}}

\newcommand{\secref}[1]{Section~\ref{#1}}
\newcommand{\appref}[1]{Appendix~\ref{#1}}
\theoremstyle{definition}

\newtheorem{theorem}{Theorem}[section]

\newtheorem{lemma}[theorem]{Lemma}
\newcommand{\bigO}{O}

\usepackage{pifont}
\algrenewcommand\algorithmicindent{1.0em} % reduce indentation
\algnewcommand\algorithmicforeach{\textbf{Parfor}}
\algdef{Sl}[FOR]{ParFor}[1]{\algorithmicforeach\ #1\ \algorithmicdo}

\algnewcommand\algorithmicparforeach{\textbf{parallel\_for\_each}}
\algdef{Sl}[FOR]{ParForEach}[1]{\algorithmicparforeach\ #1\ \algorithmicdo}

\newcommand{\SR}[1]{{#1}}
\newcommand{\CR}[1]{{#1}}
\newcommand{\bufGraph}{ForkGraph}
\newcommand{\forkP}{fork-processing pattern}
\newcommand{\bufModel}{buffered execution model}

\newcommand{\ClgC}{blue}

\newcommand{\CgmC}{green}
\newcommand{\Cgt}{magenta}
\newcommand{\CgtC}{red}

\begin{document}
\fancyhead{}

%%
%% The "title" command has an optional parameter,
%% allowing the author to define a "short title" to be used in page headers.

\def\complete{}
\ifdefined\complete{}
\title{Cache-Efficient Fork-Processing Patterns on Large Graphs (complete version)}
\else
\title{Cache-Efficient Fork-Processing Patterns on Large Graphs}
\fi

%%
%% The "author" command and its associated commands are used to define
%% the authors and their affiliations.
%% Of note is the shared affiliation of the first two authors, and the
%% "authornote" and "authornotemark" commands
%% used to denote shared contribution to the research.
\author{Shengliang Lu}
\affiliation{%
  \institution{National University of Singapore}
  \country{Singapore}}

\author{Shixuan Sun}
\affiliation{%
  \institution{National University of Singapore}
  \country{Singapore}}

\author{Johns Paul}
\affiliation{%
  \institution{National University of Singapore}
  \country{Singapore}}

\author{Yuchen Li}
\affiliation{%
  \institution{Singapore Management University}
  \country{Singapore}}

\author{Bingsheng He}
\affiliation{%
  \institution{National University of Singapore}
  \country{Singapore}}
%%
%% By default, the full list of authors will be used in the page
%% headers. Often, this list is too long, and will overlap
%% other information printed in the page headers. This command allows
%% the author to define a more concise list
%% of authors' names for this purpose.
\renewcommand{\shortauthors}{Lu et al.}

%bufGraph	1.00	1.08	1.15	1.22	1.34	1.28	1.32	1.24	1.34	1.43

%bufGraph	1.00	1.08	1.15	1.22	1.34	1.28	1.32	1.24	1.34	1.43

% \begin{filecontents*}[overwrite]{plot/CC.csv}
% dataset	index	Ligra-S	Ligra-C	Gemini-S	Gemini-C	bufGraph
% Ca	1	1.4	1	1.1	0.9	0.1
% Us	2	1.3	1	1.1	0.8	0.2
% Lj	3	1.5	1	1.2	0.9	0.2
% Wk	4	1.8	1	1.1	0.9	0.1
% Tw	5	2	1	1	0.7	0.3
% \end{filecontents*}

\begin{abstract}
As large graph processing emerges, we observe a costly \emph{\forkP{}} (FPP) that is common in many graph algorithms. The unique feature of the FPP is that it launches many \emph{independent queries} from different source vertices on the same graph. For example, an algorithm in analyzing the network community profile can execute Personalized PageRanks that start from tens of thousands of source vertices at the same time. We study the efficiency of handling FPPs in state-of-the-art graph processing systems on multi-core architectures, including Ligra, Gemini, and GraphIt. We find that those systems suffer from severe cache miss penalty because of the irregular and uncoordinated memory accesses in processing FPPs.

In this paper, we propose \emph{\bufGraph{}}, a cache-efficient FPP processing system on multi-core architectures. In order to improve the cache reuse, we divide the graph into partitions each sized of LLC (last-level cache) capacity, and the queries in an FPP are buffered and executed on the partition basis. We further develop efficient intra- and inter-partition execution strategies for efficiency. For intra-partition processing, since the graph partition fits into LLC, we propose to execute each graph query with efficient sequential algorithms (in contrast with parallel algorithms in existing parallel graph processing systems) and present an atomic-free query processing method by consolidating contending operations to cache-resident graph partition. For inter-partition processing, we propose two designs, yielding and priority-based scheduling, to reduce redundant work in processing. Besides, we theoretically prove that \bufGraph{} performs the same amount of work, to within a constant factor, as the fastest known sequential algorithms in FPP queries processing, which is work efficient. Our evaluations on real-world graphs show that \bufGraph{} significantly outperforms state-of-the-art graph processing systems (including Ligra, Gemini, and GraphIt) with two orders of magnitude speedups.

\end{abstract}

\maketitle

\section{Introduction}\label{sec:introduction}
%%%%%%%%%%%%%%%%%%%%%%%%%%%%%%%%%
%%  Applications + fork executation
%%%%%%%%%%%%%%%%%%%%%%%%%%%%%%%%%
Graphs are \emph{de facto} data structures in various applications such as social network analysis, bioinformatics, online transaction analysis, and weblink analysis. We observe a costly \emph{\forkP{}} (FPP) \CR{that is common} in many graph processing algorithms, as defined in Algorithm~\ref{algo:fork}. The unique feature of an FPP is that it launches many \emph{independent queries} from different source vertices on the same graph (we call those queries \emph{FPP queries}). Below are several representative examples of FPP-based graph algorithms.

\setlength{\textfloatsep}{1pt}% Remove \textfloatsep
\begin{algorithm}[h]
    \small
    \caption{Fork-processing pattern (FPP) on graph.}
    \label{algo:fork}
    \begin{algorithmic}[1]
        \State Generate vertex set $S$ \label{fork:pre}
        \ParForEach{vertex $v$ $\in$ $S$}
            \State Launch a graph query from $v$ \label{fork:cgq}
        \EndFor
    \end{algorithmic}
\end{algorithm}

%%%%%%%%%%%%%%%%%%%%%%%%%%%%%%%%%
%%  Different schemes' profiling
%%%%%%%%%%%%%%%%%%%%%%%%%%%%%%%%%
\setlength{\textfloatsep}{1pt}
\begin{table*}[t]
    \captionsetup{aboveskip=0pt}
    %\captionsetup{belowskip=-10pt}
    \caption{\SR{Profiling performance analysis of processing 10,000 PPRs on LiveJournal graph using existing GPSs.}}
    \label{tab:intro}
    %\begin{table}[t]
%\centering
%\caption{Graph inputs}
%\label{tab:dataset}
%\begin{tabular}{r|r|r}
%Data sets & Vertices & Num. Edges\\ \hline
%CAL~\cite{demetrescu20089th} & 1.9M   & 4.7M      \\ \hline
%USA~\cite{demetrescu20089th} & 23.9M    & 58.3M       \\ \hline
%LiveJ~\cite{snapnets} & 4.8M    & 68.5M     \\ \hline
%Wiki~\cite{davis2011university}  & 16.8M & 99.8M \\ \hline
%UK-2002~\cite{BCSU3}  & 18.4M & 261.8M \\ \hline
%Twitter~\cite{Kwak10www}  & 41.7M    & 1.5B      \\ \hline
%Friendster~\cite{snapnets}  & 68.3M    & 2.6B      \\ \hline
%\end{tabular}
%\end{table}

%\footnotesize
\centering
%\resizebox{\textwidth}{!}{
  \begin{tabularx}{\textwidth}{|l|lXX|lXX|lXX|}
    \Xhline{2\arrayrulewidth}
    System             & \multicolumn{3}{c|}{Ligra}   & \multicolumn{3}{c|}{Gemini}  & \multicolumn{3}{c}{GraphIt} \\ \hline
    \#Threads in total      & 1               & 10     & 10    & 1               & 10     & 10    & 1               & 10     & 10 \\ \hline
    Execution Scheme        & single-threaded & $t=10$ & $t=1$ & single-threaded & $t=10$ & $t=1$ & single-threaded & $t=10$ & $t=1$ \\ \Xhline{2\arrayrulewidth}
    Instructions ($\times 10^{14}$)     & 4.57     & 4.59   & 4.56    & 2.07     & 2.20   & 2.46    & 1.30     & 1.55   & 1.31 \\
    %LLC loads ($\times 10^{11}$)        & 90.97    & 89.95  & 92.06   & 13.02    & 14.63  & 13.65   & 15.86    & 16.25  & 16.03 \\
    LLC loads ($\times 10^{12}$)        & 9.10     & 9.00   & 9.21    & 1.30     & 1.46   & 1.37    & 1.59     & 1.63   & 1.60 \\
    %IPC                     & 0.69     & 0.63   & 0.64     & 1.27     & 0.7    & 0.97    & 1.16     & 0.6    & 1       \\
    %CPU utilization         & 100.0\%  & 79.2\% & 94.7\%  & 100.0\%  & 95.0\% & 99.0\%  & 100.0\%  & 87.4\% & 99.9\%  \\
    LLC miss ratio                      & 50.0\%   & 48.1\% & 79.0\%  & 40.1\%   & 31.6\% & 76.4\%  & 50.1\%   & 38.9\% & 85.6\% \\
    %Runtime (s)                         & 168247   & 27527  & 24301   & 41962    & 9202   & 5897    & 30197    & 6669   & 5734   \\
    Runtime (hour)                      & 46.74    & 7.65   & 6.75    & 11.66    & 2.56   & 1.64    & 8.39     & 2.09   & 1.59 \\
    \Xhline{2\arrayrulewidth}
\end{tabularx}
%}
\vspace{-13pt}
\end{table*}

\begin{enumerate}[wide,noitemsep,topsep=1pt]
\item \emph{Betweenness centrality} (BC) is widely used to calculate the relative importance of vertices in a graph~\cite{jamour2017parallel}. On an unweighted graph, BC is solved by first invoking many independent BFSs (breadth-first searches), each from a random vertex. Next, the algorithm gathers the results of each BFS to obtain the centrality of vertices~\cite{brandes2001faster}. Although various algorithm variants have been proposed, they have common FPPs of launching massive BFS queries~\cite{wang2001fast,gera2020traversing}.

\item \emph{Network community profile} (NCP) is defined as the function of the (approximate) best conductance for clusters of a given size in the graph versus the cluster size \cite{leskovec2009community}. An efficient method computing NCP is based on local clustering algorithms, which start a number of PPRs (personalized page ranks) from randomly selected vertices to calculate NCP approximately~\cite{shun2016parallel, wang2017capacity, yang2015defining, fortunato2016community}. The number of PPRs can be at the scale of tens of thousands in the previous study~\cite{shun2016parallel}.

\item \emph{Landmark labeling} (LL) pre-computes the shortest paths between selected landmark vertices to accelerate the path queries. Researchers proposed to compute the labels by executing a batch of SSSPs (single-source shortest paths) or BFSs simultaneously \cite{akiba2013fast}. The number of queries in a batch can range from 16 to 1,024 in the previous studies~\cite{akiba2013fast}.
\end{enumerate}

In practice, the processing time of the FPP is the major bottleneck of those graph algorithms, which takes an overwhelming majority of the execution time ($>90\%$) in our experiments. In this paper, we study whether and how we can improve the performance of handling FPPs on large graphs.

%%%%%%%%%%%%%%%%%%%%%%%%%%%%%%%%%
%%  Motivation
%%%%%%%%%%%%%%%%%%%%%%%%%%%%%%%%%
As large graph processing emerges recently, substantial efforts have been made in developing parallel graph processing systems (GPSs)~\cite{shun2013ligra, nguyen2013lightweight,zhu2016gemini,zhang2018graphit}. Those GPSs mainly focus on improving the performance of a single query by taking advantage of the intra-query parallelism. Since an FPP has many independent queries, we study how existing GPSs can take advantage of inter-query parallelism. To this end, we use $t$ to denote the number of threads assigned to a query, and evaluate different $t$ values for balancing the intra- and inter-query parallelisms.

We evaluate three state-of-the-art GPSs (Ligra~\cite{shun2013ligra}, Gemini~\cite{zhu2016gemini}, and GraphIt~\cite{zhang2018graphit}) on a 10-core machine (with hyperthreading disabled). Table~\ref{tab:intro} presents the performance analysis of handling \SR{10,000 PPRs} for NCP on a real graph. The detailed experimental setup can be found in \secref{sec:exp}. For varying different $t$ values, we fix the total number of threads to be ten (one thread per core). Specifically, when $t=1$, GPSs fully take advantage of inter-query parallelism. When $t=10$, GPSs process queries one by one, and each query runs in parallel, \CR{taking advantage of intra-query parallelism only}. We find that the configuration of $t=1$ achieves the best performance among different $t$ values. In the table, we also show the profiling of executing GPSs using a single thread.

We make an important observation across different GPSs. Despite that \CR{the executions with configurations of $t=1$ achieve the best performance by taking advantage of inter-query parallelism, they suffer} from a high LLC (last level cache) miss ratio. It can be up to $85.6\%$, a huge rise from both the single-threaded execution and the approach of intra-query parallelism ($t=10$). When $t=1$, threads are handling FPP queries independently and simultaneously, and they are filling up the CPU cache with different parts of the graph. Such uncoordinated memory accesses among FPP queries cause severe LLC cache misses. We present more details in \secref{sec:background}.

%%%%%%%%%%%%%%%%%%%%%%%%%%%%%%%%%
%%  Systems
%%%%%%%%%%%%%%%%%%%%%%%%%%%%%%%%%
To improve the efficiency of handling \CR{FPPs}, we develop \emph{\bufGraph{}}, a cache-efficient system for processing FPPs for in-memory graphs on multi-core machines. The core design of \bufGraph{} is based on a novel \emph{\bufModel{}} on graphs to leverage the locality and sharing opportunities among FPP queries by coordinating their operations to the graph. Specifically, we divide the graph data into LLC-sized partitions and associate each partition with a buffer to store the queries' operations to the partition. We dynamically schedule a partition to process, and the buffered operations are performed in a batch on the cache-resident graph partition.

We further develop efficient intra- and inter-partition processing strategies for efficiency. For intra-partition processing, since the graph partition fits into LLC, we propose to execute each \CR{buffered operation} with efficient sequential algorithms (unlike parallel algorithms in existing graph systems that have more work due to parallelism) and develop atomic-free mechanisms by consolidating contending operations \CR{in the} cache-resident graph partition. For inter-partition processing, we observe that a wrong execution order of graph partitions causes a significant amount of redundant work, as well as the benefits of each buffered operation in the same graph partition vary significantly in many graph applications. Thus, we propose two designs accordingly. The first is priority-based scheduling to select the partition that could lead to convergence quickly. The second is yielding optimization that early terminates a query's intra-partition processing to reduce redundant work.

Besides, with the designs in intra- and inter-partition processing, \bufGraph{} performs the same amount of work, to within a constant factor, as the fastest known sequential algorithms in FPP queries processing, which is work efficient.

We perform a comprehensive analysis of \bufGraph{}'s performance in comparison with three state-of-the-art GPSs (Ligra, Gemini, and GraphIt). The workload includes three applications BC, NCP, and LL on eight real-world graphs. Our experiments show that \bufGraph{} reduces the total number of LLC misses by up to a factor of $100\times$ compared to other GPSs, and consistently outperforms GPSs, showing $32\times$, $307\times$, and $38\times$ speedups over Ligra, Gemini, and GraphIt on average, respectively.

% \vspace{1mm}\noindent
% \textbf{Contributions.}
% The main contributions of this paper are as follows.
% \begin{itemize}[wide,noitemsep,topsep=2pt]
%     \item We observe the \forkP{} that is common and costly in many graph processing algorithms.
%     \item We propose a cache-efficient FPP queries processing system named \bufGraph{}. It embraces a novel buffered execution model and work-efficient intra- and inter-partition handling, which is proved to be work efficient.
%     \item We perform a comprehensive evaluation of \bufGraph{}'s performance and show that \bufGraph{} achieves much higher cache efficiency and significantly outperforms state-of-the-art approaches with two orders of magnitude speedups over them on average.
% \end{itemize}
\CR{
Supplement results are presented
\ifdefined\complete{}
in the appendix.
\else
in the complete version~\cite{forkgraph2021}.
\fi
Our source code is publicly available at (\emph{\url{https://github.com/Xtra-Computing/ForkGraph}}).}

The remainder of this paper is organized as follows. In \secref{sec:background}, we introduce the preliminary and motivation. \secref{sec:overview} presents an overview of \bufGraph{}. We present the design of intra- and inter-partition processing in \secref{sec:cache} and \secref{sec:scheduling}, respectively. We present the experimental results in \secref{sec:exp}. Finally, we review the related work in \secref{sec:related} and conclude in \secref{sec:conclusion}.

\section{Preliminaries and Motivations} \label{sec:background}

%In this section, we first introduce the preliminaries of graphs and the FPP, followed by some details on the motivations.

\subsection{Preliminaries}
A graph $G = (V, E)$ is defined to be a directed graph, where $V$ is a set of vertices and $E$ is a set of edges. An undirected graph can be represented as a directed graph by replacing each undirected edge with two edges from both directions. Given $G = (V, E)$, graph partition is defined as follows.

\vspace{1mm}\noindent
\textbf{Definition 2.1} {(\textsc{Graph Partition})}
A partition plan of graph $G$ is a division of $V$ into $|\bm{P}|$ disjoint vertex sets. A partition $P_i$ of the graph $G$ is a subgraph of $G$ on the $i$-th vertex set with $1 \leq i \leq |\bm{P}|$. $V_{P_i}$ and $E_{P_i}$ denote the vertex set and the edge set of $P_i$, respectively. Specifically, $E_{P_i} = \{e(u, v) \in E | u \in V_{P_i} \}$.
    %partition plan $\Phi$ satisfies that (1) given $P_i, P_j \in \bm{P}$ where $i \neq j$, $V_{P_i} \cap V_{P_j} = \emptyset$; (2) $\cup_{1 \leq i \leq |\bm{P}|} V_{P_i} = V$; and (3) given $P_i \in \bm{P}$,
%\end{definition}

\vspace{1mm}\noindent
\textbf{FPP-based graph applications.} We have observed the existence of FPPs in many graph applications beyond those presented in the Introduction\CR{, including }ant colony optimization \cite{dorigo2006ant}, single-source $k$-shortest path~\cite{yen1970algorithm}, and graph learning using random walks \cite{perozzi2014deepwalk, grover2016node2vec}.

\vspace{1mm}\noindent
\textbf{Definition 2.2} {(\textsc{FPP Queries})}
Given a graph $G = (V, E)$, the FPP queries denoted by $Q = \{q_1, q_2,...,q_{|Q|}\}$ are a set of graph queries that are homogeneous graph queries (e.g., PPRs) \emph{simultaneously} launched from source vertices $src_1, src_2,...,src_{|Q|} \in V$ on the same graph $G$, where $|Q|$ denotes the number of queries in the FPP.

$|Q|$ varies in different applications. In the previous studies, it ranges from tens to thousands. For example, Shun et al.~\cite{shun2016parallel} run $10^5$ PPRs from random vertices to compute NCP. The LL calculation in Akiba et al.~\cite{akiba2013fast} performs batches of SSSPs/BFSs simultaneously, and the number of queries in a batch varies from 16 to 1,024. \SR{Some graph workloads in the previous studies do not belong to FPP, such as single-query applications like a single BFS~\cite{beamer2012direction}, multiple dependent queries like multi-commodity flow~\cite{awerbuch1994improved}, and heterogeneous graph queries~\cite{yen1971finding}. As FPP is emerging in many graph processing, mining, and learning applications, we focus on FPP queries.}

% \SR{Non-FPP queries include single-query applications like a single BFS, multiple, dependent queries like multi-commodity flow~\cite{awerbuch1994improved}, and heterogeneous graph queries. The acceleration of single queries has been overwhelmingly studied. Besides, to the best of our knowledge, there are few algorithms that are decomposed into heterogeneous or dependent queries. Therefore, in this work, we focus on efficiently handling FPP queries.}

\vspace{1mm}\noindent
\textbf{Definition 2.3} {(\textsc{An FPP Query's Operations to Graph})} An FPP query can have many operations that are spreading randomly in the graph. We define an operation of an FPP query as a triple in the format of $\langle q_i, v, val\rangle$ to represent an operation belonging to $q_i$ at vertex $v$ with value(s) $val$.

For example, in an SSSP query $q_1$, each vertex $v$ is assigned with a property that represents the length of the shortest path found from source vertex $src_1$ to $v$; an operation to vertex $v$ contains the length $l$ of a path from $src$ to $v$, denoted as $\langle q_1, v, l\rangle$. If $l$ is shorter than the existing path, we update the vertex's property using the operation and generate new operations to the vertex's neighbors.

\subsection{Motivation}
We present more detailed results of Table~\ref{tab:intro} to explain our motivation. \SR{\emph{First, the large number of LLC misses is the performance bottleneck.} When evaluating 10,000 PPRs using the three GPSs, we observed that the number of stalled memory cycles is $34-40\%$ of the total time spent in memory units when leveraging intra-query parallelism ($t=10$). The percentage increases up to $55\%$ when leveraging inter-query parallelism ($t=1$). The memory stalls are mainly caused by the LLC misses, which bottleneck the performance. Thus, we focus on LLC among the multi-layer caches in this work.}

\emph{Second, leveraging the inter-query parallelism in existing GPSs brings uncoordinated memory accesses and thus more LLC misses, which limit the potential benefits of the inter-query parallelism.} Figure~\ref{fig:intro_motivate} shows the evaluation of the GPSs' performance and the number of LLC misses by varying the $t$ value on a 10-core machine. Figure~\ref{fig:intro_content} shows that leveraging inter-query parallelism by setting $t=1$ is better than other settings, mainly because of the reduced synchronization and locking overhead among threads, as well as better load balancing among threads. However, as shown in Figure~\ref{fig:intro_llc}, the total number of LLC misses significantly increases (up to $2.1\times$) as $t$ changes from ten to one.

\SR{As threads under the inter-query parallelism are filling up the CPU cache with different parts of the graph, the uncoordinated memory accesses among FPP queries cause high LLC cache misses and limit the potential benefits of the inter-query parallelism.} We use Figure~\ref{fig:schemec} to illustrate the scenario of uncoordinated operations of FPP queries to the graph. In this example, a GPS is processing three FPP queries simultaneously, each with a thread sharing the LLC. Since threads are handling FPP queries independently, they contend with each other to the limited cache for storing different parts of the graph. This causes severe cache thrashing and downgrades the performance of handling FPPs.

We also observe that, among these three GPS, GraphIt is the only GPS with cache-optimized techniques to break the graph into LLC-sized segments to limit random accesses within the cache~\cite{zhang2018graphit, zhang2017making} and thus performs better than other GPSs when leveraging the intra-query parallelism. However, GraphIt is the most vulnerable GPS when leveraging the inter-query parallelism with the cache optimization. \SR{Particularly, GraphIt ($t=1$) spends 1.59 hours on 10,000 PPRs. Thus, the total CPU time on 10 cores is 15.9 hours, $190\%$ over the CPU time of the single-threaded counterpart (8.39 hours). Similarly, the CPU time of Ligra ($t=1$) is only $144\%$ over the CPU time of the single-threaded counterpart as it does not optimize the cache for intra-query parallelism.}

\SR{Although inter-query parallelism exhibits poor cache efficiency, it still outperforms the default intra-query parallelism in GPSs for most of the cases. The reasons are as follows. First, the inter-query parallelism inherently eliminates cost synchronizations among threads. Second, the inter-query parallelism can benefit from efficient execution by employing state-of-the-art sequential algorithms. Last but not least, as Beamer et al.~\cite{beamer2015locality} show that many parallel implementations do not fully utilize the memory bandwidth, the inter-query parallelism mitigates data dependences and increases memory-level parallelism during processing. It thus shows the potential for significant performance improvement on GPSs with current memory systems.}

In summary, the above-mentioned insights motivate us to propose a system to address the cache inefficiency in leveraging inter-query parallelism so that we can efficiently support FPP queries.

\begin{figure}[t]
	\centering
	\begin{subfigure}[t]{0.23\textwidth}
		\centering
		\begin{tikzpicture}
    \begin{axis}[
        %ylabel=CPU time,
        ylabel=Normalized execution time,
        %xlabel=number of available threads,
        width=45mm,
        height=40mm,
        ylabel near ticks,
        xlabel near ticks,
yticklabel style = {font=\footnotesize},
xticklabel style = {font=\footnotesize},
ylabel style={font=\small},
xlabel style={font=\small},
        symbolic x coords={Ligra,Gemini,GraphIt},
        ybar=0pt,bar width=4pt,
        xtick=data,
        enlarge x limits=0.15,
        %ytick distance = 25,
        ymin = 0,
        ymax = 1.1,
        ytick distance = 0.5,
        ytick = {0, 0.5, 1},
        legend columns=6,
        legend style={
            at={(0.5, 1.0)},
            anchor=south,
            draw=none,
            fill=none,
            font=\small,
        },
        extra y ticks = 1,
        extra y tick labels={},
        extra y tick style={grid=major,major grid style={dashed, draw=red}},
        legend image code/.code={\draw [#1] (0cm,-0.1cm) rectangle (0.1cm,0.1cm);},
        ]
    \addplot[draw=black,fill=blue!10] table [x=index, y=10, col sep=space] {plot/intro_cputime_ts2.csv};
    \addlegendentry{$t=10$}
    \addplot[draw=black,fill=blue!20] table [x=index, y=5, col sep=space] {plot/intro_cputime_ts2.csv};
    \addlegendentry{$t=5$}
    \addplot[draw=black,fill=blue!50] table [x=index, y=2, col sep=space] {plot/intro_cputime_ts2.csv};
    \addlegendentry{$t=2$}
    \addplot[draw=black,fill=blue] table [x=index, y=1, col sep=space] {plot/intro_cputime_ts2.csv};
    \addlegendentry{$t=1$}
    %\addplot[draw=black,fill=blue!60] table [x=index, y=1, col sep=space] {plot/intro_cputime_ts2.csv};
    %\addlegendentry{$t=1$}
    %\addplot[draw=black,fill=blue!40] table [x=index, y=2, col sep=space] {plot/intro_cputime_ts2.csv};
    %\addlegendentry{$t=2$}
    %\addplot[draw=black,fill=blue!20] table [x=index, y=5, col sep=space] {plot/intro_cputime_ts2.csv};
    %\addlegendentry{$t=5$}
    %\addplot[draw=black,fill=blue!10] table [x=index, y=10, col sep=space] {plot/intro_cputime_ts2.csv};
    %\addlegendentry{$t=10$}
    \end{axis}
\end{tikzpicture}
		\captionsetup{aboveskip=-10pt}
		%\captionsetup{belowskip=-10pt}
		\caption{Normalized execution time.}
		\label{fig:intro_content}
	\end{subfigure}\hfill%
	\begin{subfigure}[t]{0.23\textwidth}
		\centering
		\begin{tikzpicture}
    \begin{axis}[
        ylabel=Normalized \# LLC misses,
        %xlabel=number of available threads,
        width=45mm,
        height=40mm,
        ylabel near ticks,
        xlabel near ticks,
yticklabel style = {font=\footnotesize},
xticklabel style = {font=\footnotesize},
ylabel style={font=\small},
xlabel style={font=\small},
        symbolic x coords={Ligra,Gemini,GraphIt},
        ybar=0pt,bar width=4pt,
        xtick=data,
        enlarge x limits=0.15,
        %ytick distance = 25,
        ymin = 0,
        ymax = 2.2,
        ytick distance = 0.5,
        ytick = {0, 1, 2, 3},
        legend columns=6,
        legend style={
            at={(0.5, 1.0)},
            anchor=south,
            draw=none,
            fill=none,
            font=\small,
        },
        extra y ticks = 1,
        extra y tick labels={},
        extra y tick style={grid=major,major grid style={dashed, draw=red}},
        legend image code/.code={\draw [#1] (0cm,-0.1cm) rectangle (0.1cm,0.1cm);},
        ]
    \addplot[draw=black,fill=blue!10] table [x=index, y=10, col sep=space] {plot/intro_llc_ts2.csv};
    \addlegendentry{$t=10$}
    \addplot[draw=black,fill=blue!20] table [x=index, y=5, col sep=space] {plot/intro_llc_ts2.csv};
    \addlegendentry{$t=5$}
    \addplot[draw=black,fill=blue!50] table [x=index, y=2, col sep=space] {plot/intro_llc_ts2.csv};
    \addlegendentry{$t=2$}
    \addplot[draw=black,fill=blue] table [x=index, y=1, col sep=space] {plot/intro_llc_ts2.csv};
    \addlegendentry{$t=1$}
    %\addplot[draw=black,fill=blue!60] table [x=index, y=1, col sep=space] {plot/intro_llc_ts2.csv};
    %\addlegendentry{$t=1$}
    %\addplot[draw=black,fill=blue!40] table [x=index, y=2, col sep=space] {plot/intro_llc_ts2.csv};
    %\addlegendentry{$t=2$}
    %\addplot[draw=black,fill=blue!20] table [x=index, y=5, col sep=space] {plot/intro_llc_ts2.csv};
    %\addlegendentry{$t=5$}
    %\addplot[draw=black,fill=blue!10] table [x=index, y=10, col sep=space] {plot/intro_llc_ts2.csv};
    %\addlegendentry{$t=10$}
    \end{axis}

\end{tikzpicture}
		\captionsetup{aboveskip=-10pt}
		%\captionsetup{belowskip=-10pt}
		\caption{Normalized \#LLC misses.}
		\label{fig:intro_llc}
	\end{subfigure}
	\captionsetup{aboveskip=1pt}
	\captionsetup{belowskip=-10pt}
	\caption{GPSs' performance affected by cache contention with different numbers of threads assigned to each query, tested with 10,000 PPRs on LiveJournal graph.}
	\label{fig:intro_motivate}
\end{figure}
\begin{figure}[t]
    \centering
    \includegraphics[width=0.3\textwidth]{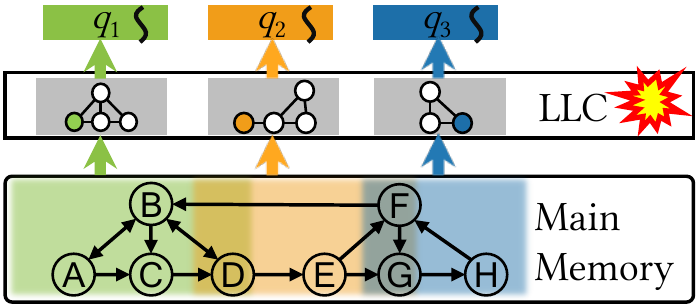}
    \captionsetup{aboveskip=2pt}
    %\captionsetup{belowskip=0pt}
    \caption{A GPS leveraging the inter-query parallelism.}
    \label{fig:schemec}
\end{figure}

\section{System Overview}\label{sec:model} \label{sec:overview}

\begin{figure}[t]
    \centering
    \includegraphics[width=0.45\textwidth]{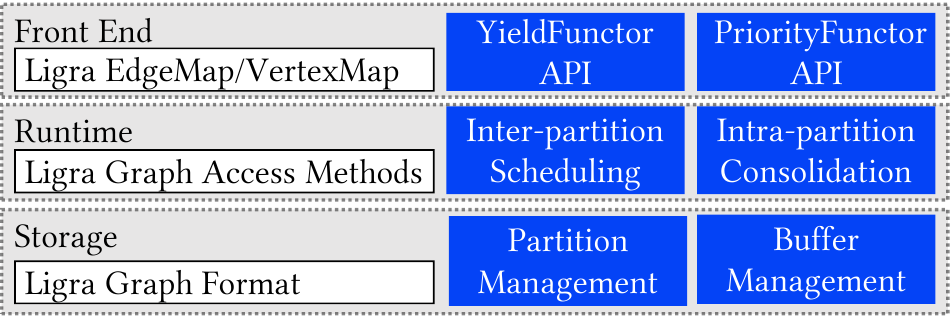}
    \captionsetup{aboveskip=2pt}
    %\captionsetup{belowskip=-13pt}
    \caption{System overview of \bufGraph{}.}
    \label{fig:sysoverview}
\end{figure}

To improve the cache efficiency of processing FPPs, we propose a cache-efficient system, namely \bufGraph{}. The core design of \bufGraph{} is based on a novel \bufModel{}. We divide graph $G$ into LLC-sized partitions and associate each partition with a buffer that stores the FPP queries' operations to the partition. The buffered execution model leverages the temporal localities among FPP queries by batching operations from different queries and executes them \CR{in a batch for each partition}. Since each graph partition can fit into LLC, random memory accesses of the operations in the batch are naturally limited in the LLC with a low cache miss rate.

We develop efficient intra- and inter-partition processing strategies for efficiency. For intra-partition processing, the work efficiency becomes essential since most operations are processed in cache-resident graph partitions. Therefore, we propose to apply sequential implementations to execute multiple operations simultaneously. We leverage inter-query parallelism by assigning a single thread to handle each buffered operation. We adopt sequential implementations because they are usually more work-efficient than parallel algorithms. Besides, \bufGraph{} consolidates the operations in the buffer that belong to the same query, and thus the operations belonging to the same query can be processed in an atomic-free manner. Moreover, this query-centric operation consolidation significantly reduces redundant operations.

For inter-partition processing, we observe that a wrong execution order of graph partitions causes a significant amount of redundant work, as well as the benefits of each buffered operation in the same graph partition vary significantly in many graph applications. We show these two observations in the experiments. Thus, we have two tasks: 1) to determine when to terminate intra-partition processing and switch to the next partition, and 2) to decide which partition to process next. We propose a yielding optimization to early terminate a query in intra-partition processing to reduce redundant work within a partition. Furthermore, \bufGraph{} applies a priority-based scheduling to select the partition that leads to convergence quickly. We will detail the design of inter-partition scheduling shortly in \secref{sec:scheduling}.

Currently, \bufGraph{} supports a list of queries commonly used in FPPs, including BFSs~\cite{shun2012brief}, DFSs (Depth-first searches)~\cite{tarjan1972depth}, SSSPs~\cite{dijkstra1959note}, PPRs~\cite{shun2016parallel}, and RWs (random walks)~\cite{alamgir2010multi}. On top of those queries, one can implement FPPs for applications like NCP, BC, and LL. Based on \bufGraph{}, users can also easily implement more \CR{fundamental} query types \CR{to support other} FPP-based graph applications.

%The buffered execution is generally applicable to many types of FPP queries, including but not limited to SSSP, BFSs, PPRs, DFSs (depth-first search), and RWs (random walks). \bufGraph{} consolidates the operations from different queries and maximizes the reuse of cached content among them. Buffering is not designed for processing a single query but significantly improves the reuse of cache contention jointly among FPP queries.

\subsection{System Architecture}
Figure~\ref{fig:sysoverview} shows the architecture of \bufGraph{}. \bufGraph{} extends the Ligra framework by including its APIs (application programming interfaces), graph access methods, and the graph storage. We choose Ligra because of its high performance and wide adoptions by lines of excellent works, including \cite{shun2016parallel,dhulipala2018theoretically, shun2020practical, dhulipala2020graph}. Another reason is that users can leverage the friendly programming interfaces in Ligra. It provides two very simple APIs \textsc{vertexMap} and \textsc{edgeMap}, used for user-defined functions over vertices and edges, respectively, making programs in Ligra very simple and concise.

\SR{On top of Ligra, we expose two user-defined APIs including priority functor and yield functor to users for customizing the logic for inter-partition scheduling. We also add the inter-partition scheduling and intra-partition consolidation on top of Ligra's runtime. We reuse the efficient graph storage and the access methods to edges/vertices in Ligra, which adopts the popular CSR (Compressed Sparse Row) format to store a graph. We add graph partition and buffer management based on the efficient storage of Ligra.}

%We expose priority and yield functors to users for customizing the logics for inter-partition scheduling. That is, \bufGraph{} uses user-defined functors to maintain a priority queue for priority-based scheduling and check the yielding condition for each query. In intra-partition processing, we enable the query-centric operation consolidation for atomic-free processing.

%\bufGraph{} reuses the efficient graph storage and the access methods to edges/vertices in Ligra, which adopts the popular CSR (Compressed Sparse Row) format to store a graph.

%\bufGraph{} uses the same format to store initial queries, operations of existing queries sent from neighboring partitions, and remaining operations after yielding as given in \secref{sec:background}. In this way, we can jointly process FPP queries at different states efficiently.

\begin{algorithm}[t]
    \small
    \caption{\bufGraph{}: FPP Processing on graph partitions $\bm{P}$.}
    \label{algo:overview}
    \begin{algorithmic}[1]
        \State \Call{\textcolor{black}{InitBuffers}}{$\bm{P}$, $Q$}\label{line:init2}
        \While{At least one buffer has operations}\label{line:fpp}
            \State {$\bm{P}_c \leftarrow$ \Call{ScheduleNextPart}{\null}}\Comment{see Sec.~\ref{sec:scheduling}}\label{line:inter}
            \State \Call{IntraPartProcess}{$\bm{P}_c$} \Comment{see Sec.~\ref{sec:cache}}\label{line:intra}
        \EndWhile
        \Procedure{ScheduleNextPart}{\null}
            \State nextP $\leftarrow$ get the next partition according to \Call{\textcolor{black}{PriorityFunctor}}{\null}
            \State \Return nextP \label{line:priority}
        \EndProcedure
        \Procedure{IntraPartProcess}{$\bm{P}_c$}
            \State \Call{Consolidate}{} operations in $\bm{P}_c$.buffer according to each query
            \ParForEach{$q$ $\in$ $Q$}\label{line:parexecute}\label{line:parfor}
                \For{op $\in$ $\bm{P}_c$.buffer.getOps($q$)}
                    \State newOps $\leftarrow$ \Call{\textcolor{black}{Compute}}{op, $\bm{P}_c$}\label{line:intra2}
                    \State Use newOps to update $\bm{P}_c$.buffer
                    \If{\Call{CanYield}{$\bm{P}_c$, $q$, \textsc{\textcolor{black}{YieldFunctor}}()}}\label{line:yield}
                        \State \Call{Yield}{\null} and goto Line~\ref{line:parfor} \Comment{terminating $q$, see Sec.~\ref{sec:scheduling}}
                    \EndIf
                \EndFor
            \EndFor
            \State Send operations to neighbor partitions \label{line:update}
        \EndProcedure
    \end{algorithmic}
\end{algorithm}

\subsection{Overall Execution Flow}
Algorithm~\ref{algo:overview} shows the overall execution flow of \bufGraph{} on handling an FPP. We assume that the graph is already partitioned, and the set of graph partitions are represented as $\bm{P}$. At Line~\ref{line:init2}, \bufGraph{} initializes \CR{a buffer, which is a dynamic-sized contiguous memory space, for each partition.} Also, we assign the FPP queries in $Q$ to the corresponding buffers in the initialization.

As long as there exists a non-empty buffer (Line~\ref{line:fpp}), \bufGraph{} invokes {\textsc{ScheduleNextPart}} to find the next partition $P_c$ to process. The scheduling is priority-based, targeting at convergence quickly. Next, \bufGraph{} processes the buffered operations in the scheduled partition $P_c$ by calling \textsc{IntraPartProcess} at Line~\ref{line:intra}. In \textsc{IntraPartProcess}, \bufGraph{} consolidates the operations in the buffer and assigns operations of the same query to a single thread so that the \textsc{Compute} procedure is in an atomic-free manner. Thus, we put a ``parallel for'' execution for different queries at Line~\ref{line:parexecute}. The processing of queries' operations can generate many new operations targeting at $P_c$ and $P_c$'s neighbor partitions. We have two cases for each of the new operations generated: 1) if it targets $\bm{P}_c$, we put it to the $\bm{P}_c$’s buffer; 2) if it targets other graph partitions, we do not send them to the buffers of their target partitions immediately. For each of the $(|\bm{P}|-1)$ partitions, we maintain a local buffer and store the operation into the local buffer. We send them in batches after finishing processing $\bm{P}_c$ (at Line~\ref{line:update}). %In our implementation, we keep the local buffer in the form of chained buckets to reduce the cache usage. When a bucket is full, we dump it with buffered writes to the main memory.

At Line~\ref{line:yield}, \bufGraph{} monitors the processing of query $q$ and adopts the yielding optimization to terminate the processing of $q$ earlier for work efficiency. The yielding happens within the intra-partition processing, but controls the amount of the work spent in the current partition to trigger the scheduling of the next partition to process. This is similar to the concept of \emph{yield} in process scheduling in operating systems. %In Line~\ref{line:update}, \bufGraph{} sends the buffered operations to the corresponding neighbor partitions along inter-partition edges.

\vspace{1mm}\noindent
\textbf{An example in buffered execution.} We use an example in Figure~\ref{fig:execution} to illustrate the buffered execution. There are 15 vertices divided into four partitions, and we set the weight of all edges to be one for simplicity purposes. We consider an \CR{SSSP-based} FPP such as LL, and use two SSSPs $q_1$ and $q_2$, from source vertices $\mathsf{A}$ and $\mathsf{I}$ in partitions $P_1$ and $P_3$, respectively. \bufGraph{} initializes these queries as operations in the corresponding buffers, as shown in Figure~\ref{fig:execution}a. Suppose $P_3$ is the first partition scheduled to process. Figure~\ref{fig:execution}b shows the state after processing $P_3$. \bufGraph{} sends operations to neighbor partitions, $P_1$ and $P_4$, as shown in their buffers $B_1$ and $B_4$. This process repeats until all buffers are empty.

\begin{figure}[t]
	\centering
	\begin{subfigure}[t]{0.44\textwidth}
		\centering
		\includegraphics[width=0.6\textwidth]{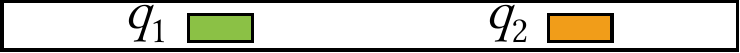}
	\end{subfigure}
	\begin{subfigure}[t]{0.236\textwidth}
		\centering
		\includegraphics[width=0.96\textwidth]{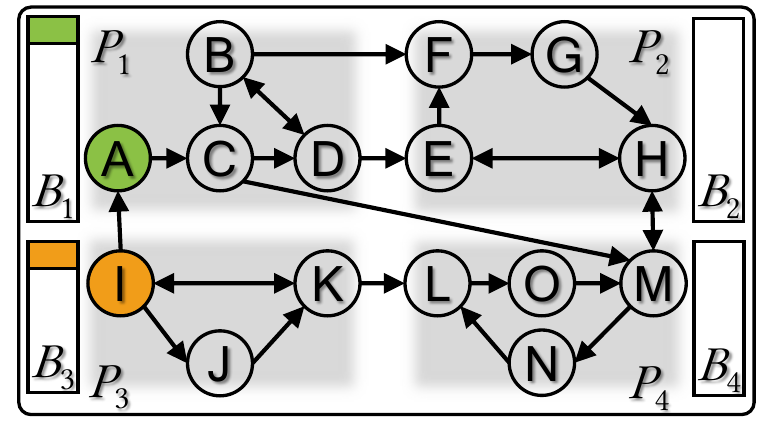}
		\captionsetup{aboveskip=0pt}
		%\captionsetup{belowskip=-5pt}
		\caption{Initialization.}
		\label{fig:exec1}
	\end{subfigure}\hfill%
	\begin{subfigure}[t]{0.236\textwidth}
		\centering
		\includegraphics[width=0.96\textwidth]{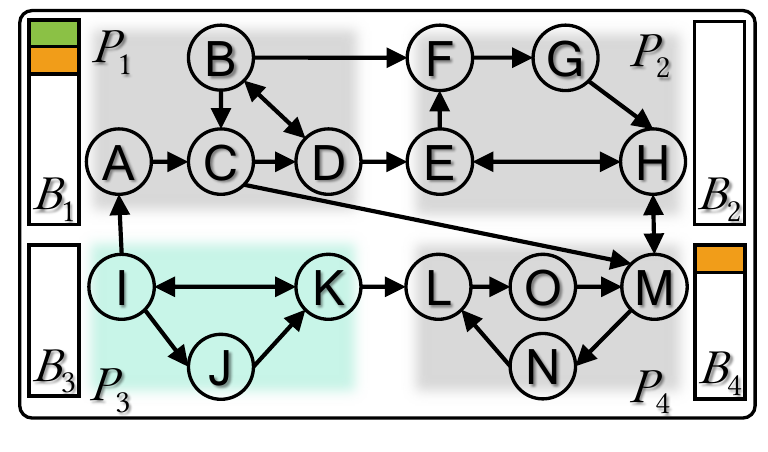}
		\captionsetup{aboveskip=0pt}
		%\captionsetup{belowskip=-5pt}
		\caption{Processed operations in $\bm{P_3}$.}
		\label{fig:exec2}
	\end{subfigure}
	\captionsetup{aboveskip=5pt}
	%\captionsetup{belowskip=0pt}
	\caption{\bufGraph{} processes two SSSP queries $\bm{q_1}$ and $\bm{q_2}$ start from vertices $\bm{\mathsf{A}}$ and $\bm{\mathsf{I}}$, respectively, as highlighted. }
	\label{fig:execution}
\end{figure}

\section{Intra-Partition Processing}\label{sec:cache}
By applying the \bufModel{}, \bufGraph{} explores the opportunities of optimizing cache-efficient \CR{intra-partition} processing of FPPs by batching the operations to LLC-sized graph partitions. In this in-cache processing, work efficiency becomes essential, with the following two main issues. First, we find that the parallel execution model adopted by the current GPSs can be extremely inefficient for in-cache processing many operations. \CR{We give more details shortly in \secref{subsec:cache:seq}.} Second, processing many operations at the same time could cause severe synchronization overheads and conflicts if they are from the same query. For example, one operation may read the property of a vertex and the other operation from the same query may update the property of the same vertex, which causes read-write conflicts. We develop an atomic-free approach to eliminate the conflicts in ~\secref{subsec:cache:consolidate}.

\subsection{Sequential vs. Parallel Execution}\label{subsec:cache:seq}

The current GPSs~\cite{shun2013ligra, nguyen2013lightweight,zhu2016gemini,zhang2018graphit} use parallel algorithms to execute one query in order to take advantage of intra-query parallelism, \CR{which, however,} is not free. Firstly, it comes with the overhead in parallelization, such as thread synchronization, locking, and scheduling. Second, most of the parallel algorithms perform significantly more work than their sequential counterparts. Those overheads can become relatively more significant given the in-cache processing.

Since there are usually more operations than the number of available CPU threads in handling FPPs, we propose to leverage the inter-query parallelism. Particularly, we choose the sequential algorithm to execute each operation and execute multiple operations simultaneously, i.e., each thread fetches one operation at a time from the buffer and processes it using the corresponding sequential algorithm. Specifically, \bufGraph{} reuses sequential algorithms from existing works. The SSSP and BFS algorithms are obtained from the problem based benchmark suite (PBBS)~\cite{shun2012brief}, and we reuse the sequential PPR implementation from \cite{shun2016parallel}.

For a partition, if the buffer only has one operation and there are no other on-going operations, we can switch to parallel algorithms to process operation. However, we hardly observe this scenario happening and thus use sequential algorithms in most cases.

%%%%%%%%%%%%%%%%%%%%%%%%%%%%%%%%%
%%  consolidation operations
%%%%%%%%%%%%%%%%%%%%%%%%%%%%%%%%%
\begin{figure}[t]
    \centering
    \begin{subfigure}[t]{0.44\textwidth}
        \centering
        \includegraphics[width=0.55\textwidth]{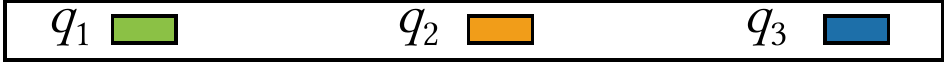}
    \end{subfigure}
    \begin{subfigure}[t]{0.22\textwidth}
        \centering
        \includegraphics[width=0.66\textwidth]{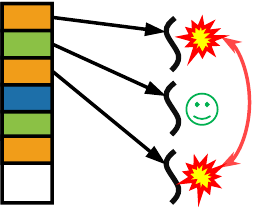}
        %\captionsetup{aboveskip=0pt}
        %\captionsetup{belowskip=-2pt}
        \caption{Without consolidation.}
        \label{fig:gather1}
    \end{subfigure}\hfill%
    \begin{subfigure}[t]{0.22\textwidth}
        \centering
        \includegraphics[width=0.88\textwidth]{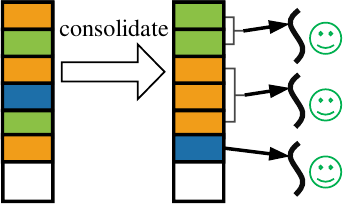}
        %\captionsetup{aboveskip=0pt}
        %\captionsetup{belowskip=-2pt}
        \caption{With consolidation.}
        \label{fig:gather2}
    \end{subfigure}
    \captionsetup{aboveskip=2pt}
    %\captionsetup{belowskip=0pt}
    \caption{Comparison of the execution on buffered operations with and without consolidation.} %Processing the unconsolidated operations in a buffer requires atomic operations to address the conflict among threads that are processing the operations belong the the same query.
    \label{fig:gather}
\end{figure}

\subsection{Query-centric Operation Consolidation}\label{subsec:cache:consolidate}

As multiple operations are processed simultaneously, different threads can process operations of the same query at the same time, which may cause access conflicts. Therefore, it requires locking and atomic operations when threads simultaneously read and write the query-specific data, e.g., the intermediate results and vertices' properties. This is also commonly seen in GPSs~\cite{shun2013ligra, nguyen2013lightweight,zhu2016gemini,zhang2018graphit} that parallelize the processing of a single query.

In this work, we propose an atomic-free approach that efficiently eliminates the conflicts among processing different operations from the same query. Particularly, \bufGraph{} \emph{consolidates} the operations based on the queries they belong to and uses a single thread to handle a set of consolidated operations from the same query. Multiple threads execute different FPP queries simultaneously, as shown in Algorithm~\ref{algo:overview} Line~\ref{line:parexecute}. The consolidation has the following benefits. First, \CR{operations of} the same query can be processed in an atomic-free manner since there is only one thread updating the query-specific data. Second, \CR{as the thread only processes data of the same query and the query-specific data is stored in a contiguous memory space, it avoids stride memory accesses to data of other queries.}

Figure~\ref{fig:gather} illustrates an example of the differences between processing with and without consolidation. We assume there are three threads. Without consolidation, threads fetch operations in the buffer and process them simultaneously. As shown in Figure~\ref{fig:gather1}, when two threads process the operations from the same query $q_2$, we need to apply atomic operations to resolve the potential conflicts (read-write conflicts). Figure~\ref{fig:gather2} shows that we perform the operation processing in an atomic-free manner by consolidating \CR{buffered operations and} assigning those operations belonging to the same query to one thread.

Given the user-defined priority functor, we can further reduce the redundant work based on query-centric consolidation. The priority functor relies on the logics of the sequential algorithms provided. During the consolidation, we maintain an order of operations from the same query according to the priority functor. In the execution, we always choose the one with the highest priority, which has more pruning power on redundant work. Particularly, commonly adopted SSSP implementations, e.g., Dijkstra's algorithm, rely on a priority queue (PQ) that assigns higher priorities to shorter paths. The algorithm uses shorter paths conveyed in the operations to prune non-optimal ones. The sequential PPR implementation from ~\cite{shun2016parallel} relies on a multiset structure that allows inserting all operations of the same query to the same queue for sequential processing in the decreasing order of the residual values. When solving a BFS, the priority value is the lowest level from the source in BFSs. %The commonly applied implementation of BFS also applies a queue to store operations by iterations.

\begin{figure}[t]
    \centering
    \begin{subfigure}[t]{0.23\textwidth}
        \centering
        \includegraphics[width=0.95\textwidth]{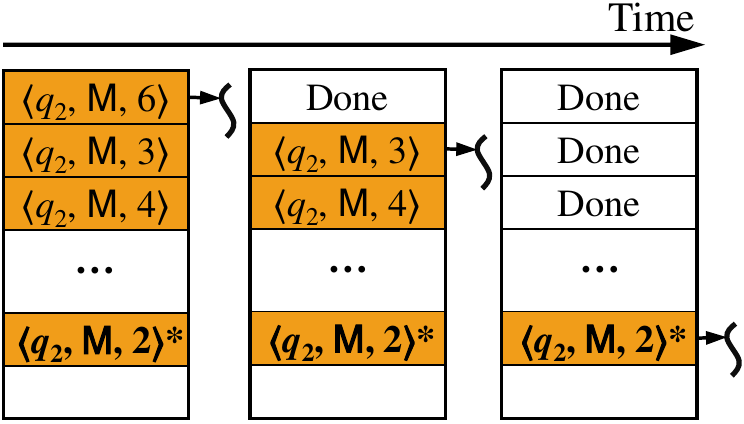}
        \captionsetup{aboveskip=2pt}
        %\captionsetup{belowskip=-2pt}
        \caption{Without priority functor.}
        \label{fig:redu1}
    \end{subfigure}\hfill%
    \begin{subfigure}[t]{0.24\textwidth}
        \centering
        \includegraphics[width=0.95\textwidth]{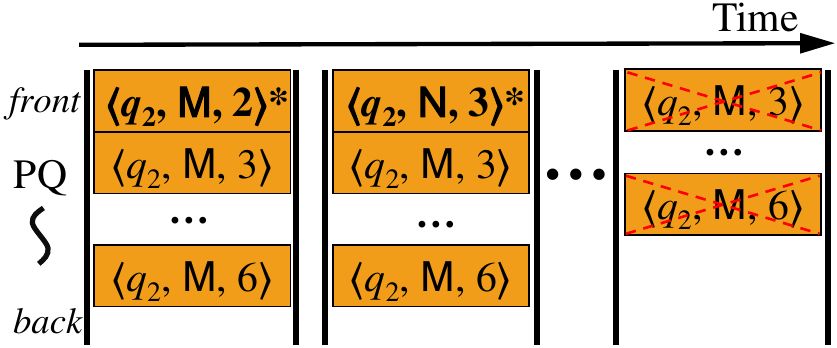}
        \captionsetup{aboveskip=2pt}
        %\captionsetup{belowskip=-2pt}
        \caption{With the priority functor.}
        \label{fig:redu2}
    \end{subfigure}
    \captionsetup{aboveskip=2pt}
    %\captionsetup{belowskip=0pt}
    \caption{Comparison of the redundancy in processing operations with and without using the priority functor in SSSP. We highlight the operations with the optimal value using $\bm{*}$.}
    \label{fig:redu}
\end{figure}

Figure~\ref{fig:redu} gives an example of processing consolidated operations of an SSSP query, where we can leverage the user-defined priority functor to reduce the number of redundant operations. Without prioritizing, a thread is fetching the operations to process one by one. The operation with the most significant value can be in any place of the buffer. For example, $\langle q_2, \mathsf{M}, 2\rangle$ contains the optimal value and is \CR{queued after} many other operations in Figure~\ref{fig:redu1}. The operations before $\langle q_2, \mathsf{M}, 2\rangle$ are all redundant. In contrast, as shown in Figure~\ref{fig:redu2}, by leveraging the priority functor of Dijkstra's algorithm for solving SSSP~\cite{dijkstra1959note}, \bufGraph{} processes the most beneficial operation of the same query. Here, the processing of $\langle q_2, \mathsf{M}, 2\rangle$ in the first place can effectively prune the redundant ones.

\section{Inter-Partition Scheduling}\label{sec:scheduling}
With efficient intra-partition processing, we still need to decide the order of scheduling partitions to execute. We observe that a wrong execution order of partitions takes more steps to converge and causes momentous redundant work. \CR{In this section}, we focus on the following two issues of inter-partition scheduling that significantly affect the work efficiency of handling the entire FPP.

\emph{First, when should \bufGraph{} terminate intra-partition processing and switch to the next partition?} One basic approach is to finish all the operations within the current partition before switching to the next partition. However, this would cause a significant amount of redundant work. According to the usage of the priority functor in query-centric consolidation, the later operations of the same query tend to have diminishing contributions to the convergence of many graph problems~\cite{shun2013ligra, nguyen2013lightweight}. They can even be pruned by the future operations generated by other partitions. Correspondingly, we present a yielding optimization to decide the early termination of a query in intra-partition processing (\secref{sec:yield}).

\emph{Second, which partition should \bufGraph{} process the next?} We propose a priority-based scheduling to select the partition that could lead to quick convergence as the next to process (\secref{sec:priority}).

\subsection{Heuristic-Based Yielding}\label{sec:yield}
The yielding optimization partially processes a partition to avoid redundant work, i.e., early termination for intra-partition processing. Determining the optimal strategy for yielding on early termination is impractical due to the complexity of graph structures and the convergence of graph applications. For example, one question here is: how to determine the utility (or benefits) of executing an operation to the convergence of the entire application. Existing studies~\cite{hassaan2011ordered, zhang2020optimizing, shun2016parallel} rely on heuristics, such as priority, to approximate the utility. Our problem is more challenging and more complex than this question. Therefore, like the previous studies, we empirically use two heuristics in \bufGraph{} to decide if a thread should yield the processing of a query. The threshold values of these heuristics can be tuned and adjusted by users. We experimentally evaluate their impacts in the evaluation.

\emph{Yielding heuristic 1: on the number of edges processed.} The first heuristic is to examine the number of edges processed since the start of processing in the partition and yield if the number is beyond a threshold. For example, the processing of a PPR gradually converges into local stable states with more edges processed but could be easily fluctuated by operations sent from other partitions in later steps. Similarly, an SSSP process goes deeply in a subgraph with more paths enumerated via processing edges, where the chance of generating non-optimal paths is higher since there are more routes from neighbor partitions that could lead to shorter ones. Such operations tend to be redundant. \bufGraph{} reduces such operations by limiting the number of edges processed.

\emph{Yielding heuristic 2: on the operations' values updated.} The second heuristic is to check if the values updated so far in the partition exceeds a range, inspired by the $\Delta$-stepping and similar approaches~\cite{meyer2003delta, blelloch2016parallel}. The heuristic works in this way: for each query in the partition, we record the value $\alpha$ conveyed from the first operation we execute. Based on the intra-partition processing, this operation has the most benefits to the convergence. If the currently processed operation's value is much worse (greater or smaller) than $\alpha$ value by a factor or a threshold, we should yield the processing.

For example, when solving SSSPs, if the shortest path to process in the current partition is greater than $dist_{min}+\Delta$, where $dist_{min}$ is the distance from the source vertex to the closest unprocessed vertices, there is a high chance that there would be better paths~\cite{meyer2003delta, zhang2020optimizing} that are not yet exploited. $\Delta$ is a tunable parameter of the $\Delta$-stepping algorithm, and the algorithm restricts the processing to vertices whose distances from the source are within $[dist_{min}, dist_{min}+\Delta)$. Similarly, in solving PPRs, if the highest residual of vertices in the partition is not significant enough, instead of continuing processing for local stable states, a better solution is to yield and wait for more influencing operations propagated to this partition.

\SR{The yielding heuristics only pause the processing of the current query in the partition. After a query yields, the unprocessed operations (including newly generated ones) for that query are kept in the partition's buffer and processed later. The processing will eventually converge as the standard algorithms and the correctness of processing results are guaranteed.}

\SR{Figure~\ref{fig:yield} shows the comparison of the execution and the number of operations during processing with and without yielding. The figure reflects the states of finishing processing $P_1$ and switching to $P_2$. In Figure~\ref{fig:yield1}, \bufGraph{} switches to $P_2$ when it finishes all operations in $P_1$ and sends three operations to $P_2$. It then schedules to process $P_2$. We can expect that \bufGraph{} has to revisit and update vertex $\mathsf{E}$ in $P_1$ again because the shortest path from $\mathsf{A}$ to $\mathsf{E}$ will be found via vertex $\mathsf{D}$.}

\SR{In Figure~\ref{fig:yield2}, \bufGraph{} yields the process after vertex $\mathsf{H}$, stores the operation at $\mathsf{H}$ in $P_1$, and sends one operation to $P_2$. When \bufGraph{} finishes the execution in $P_2$, it sends the shortest path update in an operation $\langle q_1, \mathsf{E}, 3\rangle$ to $\mathsf{E}$. Compared with the execution without yielding, the redundant operations $\langle q_1, \mathsf{D}, 5\rangle$ and $\langle q_1, \mathsf{F}, 5\rangle$ are pruned. \bufGraph{} will resume to process the remaining operation $\langle q_1, \mathsf{H}, 3\rangle$ in $P_1$'s buffer together with operations sent from $P_2$. In this way, we reduce the redundant updates and guarantee the exact results.}

%For example, Andersen et al. \cite{andersen2006local, andersen2007local} cut off the local push at a vertex $v$ if the residual is less than a threshold $d(v)\epsilon$, where $d(v)$ is the number of edges incident on $v$ and $\epsilon$ is an error parameter.

%After yielding the current query, the thread can pick up another FPP query to process; if there are no more FPP queries left in the current buffer, \bufGraph{} schedules the next partition to process. The optimization is generally applicable for a few algorithms and can also be tuned for specific optimizations.

To summarize, yielding can significantly reduce the redundant operations within the current partition and the redundant operations propagated to neighbor partitions. Furthermore, we give a theoretical proof to show that the yielding helps the work efficiency of FPP queries processing
\ifdefined\complete{}
in \appref{app:work}.
\else
in the complete version~\cite{forkgraph2021}.
\fi

\begin{figure}[t]
	\centering
	\begin{subfigure}[t]{0.225\textwidth}
		\centering
		\includegraphics[width=0.92\textwidth]{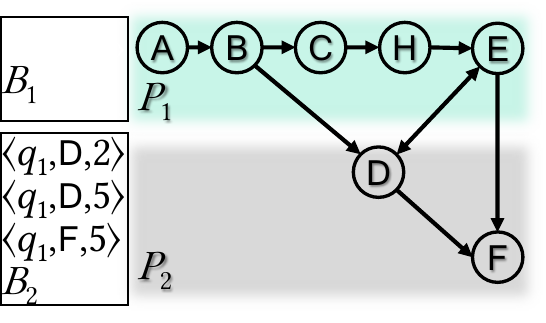}
		\captionsetup{aboveskip=0pt}
		%\captionsetup{belowskip=-5pt}
		\caption{Finish $\bm{q_1}$ in $\bm{P_1}$ w/o yielding.}
		\label{fig:yield1}
	\end{subfigure}\hfill%
	\begin{subfigure}[t]{0.235\textwidth}
		\centering
		\includegraphics[width=0.85\textwidth]{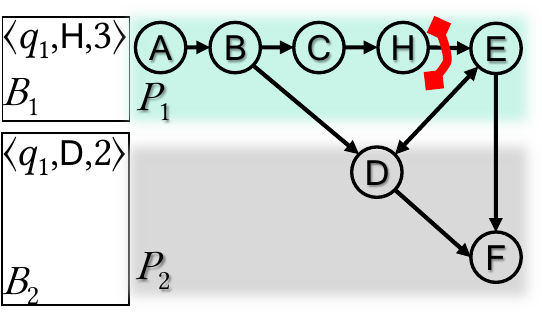}
		\captionsetup{aboveskip=0pt}
		%\captionsetup{belowskip=-5pt}
		\caption{Yield $\bm{q_1}$ at edge $\bm{\mathsf{H}}$ to $\bm{\mathsf{E}}$.}
		\label{fig:yield2}
	\end{subfigure}
	%\captionsetup{aboveskip=2pt}
	%\captionsetup{belowskip=-10pt}
	\caption{Comparison of the execution and number of operations with and without yielding in $\bm{P_1}$. Shorest path query $\bm{q_1}$ starts at vertex $\bm{\mathsf{A}}$ in $\bm{P_1}$. All edges are with unit lengths.}
	\label{fig:yield}
\end{figure}

\subsection{Priority-Based Scheduling}\label{sec:priority}
When \bufGraph{} finishes the processing of a partition, inter-partition scheduling selects another partition with a non-empty buffer to process. A wrong execution order of graph partitions leads to the repeated revisiting of partitions. To avoid such inefficiency, we propose a priority-based scheduling that aims to pick a partition that can lead to quick convergence of the FPP processing to process. We assign each partition with a priority value based on the priority functor. Intuitively, some partitions are buffering operations that would quickly lead to the convergence of queries, e.g., the shortest path or the most effective value changes in PPR updates. Therefore, we prefer to process these partitions' buffered operations than others for quick convergence.

The key question is how to determine the priority value of each partition. \SR{The priority of a partition is defined to be the highest priority value among all the operations in the partition and the priority values are generated in the priority functor.} Like many existing studies, the priority functor in \bufGraph{} is defined on per operation (i.e., per vertex), rather than on a set of vertices. For example, Dijkstra's algorithm for SSSP takes shorter distances as higher priorities~\cite{dijkstra1959note}, and it always uses the shortest path to update other vertices. As there can be many buffered operations in a partition, selecting the value with the highest priority is a simple and effective approach to determine each partition's priority.

\CR{With a given priority functor}, \bufGraph{} always schedules the partition with the highest priority in the graph to process next. However, the scheduled partition might be the most desired partition for only a subset of \CR{FPP queries}, but not for all. Redundant operations may be generated when executing operations \CR{of queries that desire other partitions}. To deal with the redundancy, we control the amount of work spent by each query using the proposed yielding optimization. We theoretically prove that \bufGraph{} is work-efficient on handling FPP queries and one of the reasons is that the yielding optimization effectively reduces such redundancy.
\ifdefined\complete{}
We present the proof in \appref{app:work}.
\else
Due to the space limit, we put the proof in the appendix of the complete version of this paper~\cite{forkgraph2021}.
\fi

\begin{figure}[t]
	\centering
	\begin{subfigure}[t]{0.42\textwidth}
		\centering
		\includegraphics[width=\textwidth]{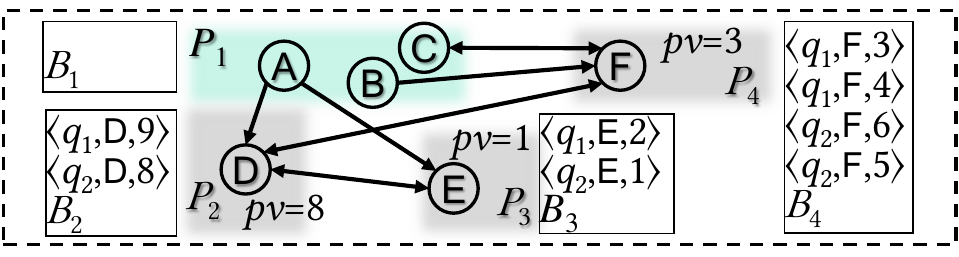}
	\end{subfigure}
	\begin{subfigure}[t]{0.42\textwidth}
		\centering
		\includegraphics[width=\textwidth]{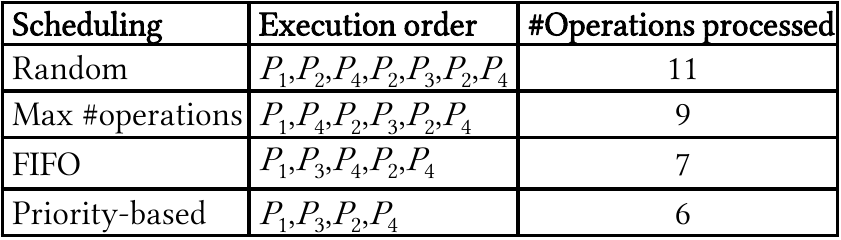}
	\end{subfigure}
	%\captionsetup{aboveskip=3pt}
	%\captionsetup{belowskip=-10pt}
	\caption{The execution orders under different scheduling methods. Shortest path queries $\bm{q_1}$ and $\bm{q_2}$ start in $\bm{P_1}$. Only vertices with edges crossing partitions are shown for brevity. All edges are with unit lengths.}
	\label{fig:priority}
\end{figure}

\SR{Figure~\ref{fig:priority} shows the comparison of the execution orders using different scheduling methods. $pv$ denotes the priority value of each partition at this stage. The random scheduling picks an arbitrary partition with operations to process at each step, which could take more steps than other methods, resulting in slow convergence. The heuristic of picking up the partition with the most number of operations (denoted as ``Max \#operations'' in the figure) could maximize the reuse of cache content. However, our experiments in \secref{subsec:eval:individual} show that it is slower than other methods because of involving more redundant work. The FIFO-based scheduling visits partitions based on the orders of operation generated. Compared with these methods, leveraging the priority functors from Dijkstra's algorithm, \bufGraph{} can schedule the execution orders to maximize the exploitation of shortest paths and reduce redundant work. In this example, the priority-based scheduling method shows the smallest number of visits to partitions and the least number of operations processed.}
\ifdefined\complete{}
\SR{The detailed work-through of this example is shown in \appref{app:walk-through}.}
\else
\SR{We provide the detailed work-through of this example in the appendix of the complete version~\cite{forkgraph2021}.}
\fi

\CR{The priority functors are programmed by users. Sometimes, it can be non-trivial to develop a functor for a certain graph operation. Fortunately, in the decades of research on graph processing, priority functors have been developed for many graph algorithms~\cite{shun2012brief, dijkstra1959note, shun2016parallel}, especially in a wide range of existing sequential algorithms. Thus, users can reuse the implementation of those priority functors as what we did in the experiments. By default, \bufGraph{} uses FIFO queues to schedule partitions.}

\section{Experimental Evaluation} \label{sec:exp}
In this section, we evaluate \bufGraph{} on handling FPPs on real-world graphs compared with state-of-the-art GPSs.

\subsection{Experimental Setup}\label{sec:exp-setup}
\textbf{Hardware configuration.}
We conduct experiments on a Linux server with a 10-core Intel\textsuperscript{\textregistered} XEON\textsuperscript{\textregistered} W-2155 CPU (hyperthreading disabled) and 256GB memory. The frequency of the CPU is 3.3GHz, and the LLC is 13.75MB. We compile all the implementations using g++ 7.5.0 with -O3 flag and OpenMP enabled.

\vspace{1mm}\noindent
\textbf{Implementation Details.}
We develop \bufGraph{} in C++ with OpenMP, where most of the components added to Ligra are developed by reusing existing GPSs' primitives or using the C++ Standard Template Library (STL). For inter-partition scheduling, we adopt the priority queue container in STL because the scheduling workload is not heavy (there are at most $|\bm{P}|$ elements maintained, each of which is the priority of a partition, where $|\bm{P}| << |V|$); as discussed early, we directly adopt the priority functors and yielding functors from the state-of-the-art sequential implementations~\cite{meyer2003delta, shun2016parallel, dijkstra1959note}. For intra-partition processing, \bufGraph{} fixes the total number of available threads as the number of hardware threads to \CR{ensure high CPU utilization, as well as} avoid high context switches and process migration overhead.

In the buffer management, we develop a simple and efficient multi-bucket structure to buffer operations, where we allocate $K$ \CR{($K \leq |Q|$)} independent buckets for each partition's buffer. We equally divide the \CR{$|Q|$} queries into $K$ corresponding groups and assign each group to explicitly use one of the buckets in each buffer, reducing the overhead of consolidating different queries' operation. In our evaluation, we set $K$ to be much larger than the number of cores in the system to allow fine-grained workload allocation. We implement the bucket using a parallel vector structure from GraphIt's code base~\cite{zhang2018graphit}, which dynamically adjusts capacities with a tunable growth factor, to make each bucket \CR{dynamic-sized and in contiguous memory}, without wasting much memory.

%\bufGraph{} relies on OpenMP's directive, ``\emph{parallel for schedule(dynamic)}``, to dynamically assign each ideal thread with an unprocessed bucket to work for efficient inter-query parallelism. In addition, the multi-bucket design also enables efficient operations sending in bulk where threads only manipulate buckets they have been assigned, resulting little interference among each others.

%We adopt the concept of having LLC-sized partitions proposed by \cite{zhou2003buffering, zhang2018graphit, zhang2017making}. This is also motivated by our profiling results: FPP processing has very severe LLC misses due to its irregular memory accesses and cache thrashing. Specifically,

For graph partitioning, we mostly use METIS~\cite{karypis1998fast}, one of the state-of-the-art edge-cut tools, to pre-process the graph with the objective of minimizing the total number of edges across different partitions. Since METIS shows poor partitioning quality on large-scale social network graphs~\cite{wei2016speedup}, we randomly partition these graphs \CR{into parts with equal number of vertices}.

\SR{In the evaluation, we use the priority \CR{functor} in Dijkstra's algorithm~\cite{dijkstra1959note}, and set yielding heuristic 2 as guided in the $\Delta$-stepping algorithm~\cite{meyer2003delta} for BC and LL; we use the priority \CR{functor} in~\cite{shun2016parallel},}
\ifdefined\complete{}
\SR{and set yielding heuristic 1 as guided in \appref{app:work} for NCP.}
\else
\SR{and set yielding heuristic 1 as guided in the appendix of \cite{forkgraph2021} for NCP.}
\fi
\SR{Particularly, the \CR{priority-based scheduling} is implemented as priority queue with binary comparison functors (comparators) for all the three applications evaluated in this paper. In the evaluation of BC and LL, the functors compare two operations and return ``true'' if the path length carried by the first operation is shorter than that carried by the second. For NCP, the functor returns ``true'' if the residual carried by the first operation is higher than the one carried by the second. This setup is used for all the tests in the rest of this paper, unless specified otherwise.}

\vspace{1mm}\noindent
\textbf{System settings.}
\SR{The system settings contain three major \CR{perspectives}. First, given a machine, for any given graph, the partition size is fixed to be the LLC size. In other words, $|\bm{P}|$ is calculated as $G.size/LLC.size$. Second, users can provide the priority functors. Typically, users obtain the priority functors from existing sequential algorithms. If there is no functor provided, \bufGraph{} uses a FIFO queue for processing and scheduling by default. Third, we set the yielding heuristics based on the work efficiency analysis. \CR{In this way,} \bufGraph{} shows comparable or the best performance among different settings, \CR{as shown later in Table~\ref{tab:priority}}. Although \CR{these settings} may not achieve the best-case performance, they are sufficiently good in practice.}

\vspace{1mm}\noindent
\textbf{Comparisons.} We compare \bufGraph{} to the other three representative GPSs, Ligra~\cite{shun2013ligra}, Gemini~\cite{zhu2016gemini}, and GraphIt~\cite{zhang2018graphit}. Ligra has the fastest implementation of many algorithms, as it is still actively maintained~\cite{shun2016parallel, dhulipala2018theoretically}. Gemini is a distributed graph processing system with notable shared-machine performance. Gemini is compiled and tested with message passing functions disabled. We also compare with GraphIt~\cite{zhang2018graphit}, the state-of-the-art DSL (domain-specific language) for high-performance graph analytics. All three systems allow users to explore various optimization and tradeoffs (such as push vs. pull, and dense vs. sparse frontier). \SR{In our evaluation, all the systems are carefully tuned and tested with different configurations. Particularly, the tuned configurations include the direction switch threshold in Ligra and Gemini, the scheduling in GraphIt, the yielding and priority-based scheduling in \bufGraph{}. We present the best result for all systems among those tests.}

As \CR{mentioned} in the Introduction, we use $t$ to denote the number of threads assigned to a query in Ligra, Gemini, and GraphIt. Particularly, when $t=10$ and FPP queries are executed one by one (with intra-query parallelism \CR{on ten cores}), we denote these three GPSs as Ligra ($t=10$), Gemini ($t=10$), and GraphIt ($t=10$), respectively. When $t=1$, we run each query independently using one hardware thread and use OpenMP's dynamic scheduling to allow each GPS to fetch and process FPP queries without synchronizations. We denote those \CR{implemenations} as Ligra ($t=1$), Gemini ($t=1$), and GraphIt ($t=1$), respectively.

We measure the time for each system to complete the FPP queries processing. This measurement excludes the time spent reading data from the disk because we focus on optimizing the in-memory computation. For each of the FPP-based graph applications evaluated, we generate three testing sets per graph, representing three testing instances of an application on a graph. We run all tests five times and report the average execution time of the $3\times5=15$ tests. Note that there are a few FPP queries in each test, representing a single instance of an FPP application. Details are given below.

\vspace{1mm}\noindent
\textbf{Applications.} We evaluate competing systems' performance on three applications, BC, LL, and NCP. To keep consistent with previous work~\cite{gera2020traversing, shun2016parallel, akiba2013fast}, we configure the three applications as follows.
\begin{itemize}[wide,noitemsep,topsep=0pt,leftmargin=0pt]
    \item The original \emph{BC} performs one SSSP from each vertex for a weighted graph (for unweighted graphs, one BFS from each vertex). This is too time-consuming to be practical. Instead, we adopt an approximate approach by Eppstein et al.~\cite{wang2001fast}. The algorithm samples the starting nodes from the graph. Therefore, in our evaluation, we randomly sample a batch of $100$ source vertices for each graph~\cite{gera2020traversing}.
    \item \SR{The testing of \emph{NCP} follows \cite{shun2016parallel}. Each NCP requires running PPRs using a seeding of $0.01\%$ of the vertices that are randomly sampled in the target graph.}
    \item \SR{The testing of \emph{LL} follows \cite{akiba2013fast}; each LL requires executing 1,024 independent SSSPs from source vertices that are randomly sampled in the target graph.}
\end{itemize}

\vspace{1mm}\noindent
\textbf{Data sets.} All the data sets are publicly available and widely used in the previous literature~\cite{zhang2018graphit, shun2013ligra, zhu2016gemini} to benchmark algorithms and frameworks. The data sets are listed in Table~\ref{tab:datasets}. Following the experimental studies in \cite{shun2013ligra} and \cite{dhulipala2017julienne}, we create weighted graphs by selecting edge weights between $[1, log|V|)$ uniformly at random. \texttt{Ca}, \texttt{Us}, and \texttt{Eu} are road networks, \texttt{Or}, \texttt{Lj}, and \texttt{Tw} are social networks, \texttt{Wk} is a hyperlink network, and \texttt{Pt} is a citation network.

\vspace{1mm}\noindent
\textbf{Experimental outline.} We first present the overall comparison between \bufGraph{} and other GPSs in \secref{subsec:eval:overall}. In \secref{subsec:eval:cache}, we present the profiling results on cache performance and work efficiency. In \secref{subsec:eval:individual}, we evaluate the impacts of individual techniques and system tuning. \CR{Due to the space limit,} we summarize some other results in \secref{subsec:eval:other} and present the details
\ifdefined\complete{}
in the appendix.
\else
in the appendix of the complete version~\cite{forkgraph2021}.
\fi

%%%%%%%%%%%%%%%%%%%%%%%%%%%%%%%%%
%%  data set
%%%%%%%%%%%%%%%%%%%%%%%%%%%%%%%%%
\setlength{\textfloatsep}{1pt}
\begin{table}[t]
    %\captionsetup{aboveskip=0pt}
    %\captionsetup{belowskip=-5pt}
    \caption{Input graphs ($\overline{d}$ is the average degree).}
    \label{tab:datasets}
    %\begin{table}[t]
%\centering
%\caption{Graph inputs}
%\label{tab:dataset}
%\begin{tabular}{r|r|r}
%Data sets & Vertices & Num. Edges\\ \hline
%CAL~\cite{demetrescu20089th} & 1.9M   & 4.7M      \\ \hline
%USA~\cite{demetrescu20089th} & 23.9M    & 58.3M       \\ \hline
%LiveJ~\cite{snapnets} & 4.8M    & 68.5M     \\ \hline
%Wiki~\cite{davis2011university}  & 16.8M & 99.8M \\ \hline
%UK-2002~\cite{BCSU3}  & 18.4M & 261.8M \\ \hline
%Twitter~\cite{Kwak10www}  & 41.7M    & 1.5B      \\ \hline
%Friendster~\cite{snapnets}  & 68.3M    & 2.6B      \\ \hline
%\end{tabular}
%\end{table}

\footnotesize
\centering
\resizebox{0.48\textwidth}{!}{
  \begin{tabular}{|crrrrrr|}
    \Xhline{2\arrayrulewidth}
    \bf{Graph} & \bf{Source} & \multicolumn{1}{c}{$\bm{\#V}$} & \multicolumn{1}{c}{$\bm{\#E}$} & \multicolumn{1}{c}{{$\bm{\overline{d}}$}} & \multicolumn{1}{c}{\SR{\bf{Memory}}} & \multicolumn{1}{c|}{$\bm{|P|}$}\\
    \hline
    \texttt{Ca} & {California}~\cite{demetrescu20089th} & 1.9M & 4.6M & 2.4 & 0.07GB & 5\\
    \texttt{Us} & {USA}~\cite{demetrescu20089th} & 23.9M & 57.7M & 2.4 & 0.82GB & 62\\
    \texttt{Eu} & {Europe}~\cite{bader201110th} & 50.9M & 0.1B & 2.1 & 1.65GB & 120\\
    \texttt{Or} & {Orkut}~\cite{snapnets} & 3.1M &0.1B & 38.1 & 1.37GB & 100\\
    \texttt{Wk} & {Wikipedia}~\cite{davis2011university} & 3.6M & 45.0M & 12.6 & 0.54GB & 40\\
    \texttt{Lj} & {LiveJournal}~\cite{snapnets} & 4.8M & 87.5M & 18.0 & 1.04GB & 76\\
    \texttt{Pt} & {Patents}~\cite{snapnets} & 16.5M & 33.0M & 2.0 & 0.50GB & 37\\
    \texttt{Tw} & {Twitter}~\cite{Kwak10www} & 61.6M & 1.5B & 23.8 & 17.27GB & 1256\\
    %\texttt{Fr} & {Friendster}~\cite{snapnets} & 21,297,772 & 530,051,090 & 25\\
    %\texttt{Pt}& {ProteinDB}~\cite{DistGraph}  & 48,748,701 & 387,730,070 & 8\\
    \Xhline{2\arrayrulewidth}
  \end{tabular}
}

\end{table}

%%%%%%%%%%%%%%%%%%%%%%%%%%%%%%%%%
%%  overall
%%%%%%%%%%%%%%%%%%%%%%%%%%%%%%%%%
\subsection{Overall Performance Comparison}\label{subsec:eval:overall}
%%%%%%%%%%%%%%%%%%%%%%%%%%%%%%%%%
%%  overall
%%%%%%%%%%%%%%%%%%%%%%%%%%%%%%%%%
\begin{figure}[t]
    \centering
    \begin{subfigure}[t]{0.49\textwidth}
        \centering
        \begin{tikzpicture}
    \begin{axis}[
        legend cell align={left},
        y tick label style={
            /pgf/number format/fixed,
            /pgf/number format/fixed zerofill,
            /pgf/number format/precision=1,
        },
yticklabel style = {font=\footnotesize},
xticklabel style = {font=\footnotesize},
ylabel style={font=\footnotesize},
xlabel style={font=\footnotesize},
        ytick distance = 0.5,
        ymax = 3,
        ymin = 0,
        enlarge x limits=0.15,
        width=84mm,
        height=36mm,
        ylabel=Normalized execution time,
        %xlabel=data sets,
        ybar=0pt,
        bar width=3pt,
        xtick = {1, 2, 3, 4, 5, 6, 7, 8},
        xticklabels = {\texttt{Ca}, \texttt{Us}, \texttt{Eu}, \texttt{Or}, \texttt{Wk}, \texttt{Lj}, \texttt{Pt}, \texttt{Tw}},
        extra y ticks = 1,
        extra y tick labels={},
        extra y tick style={grid=major,major grid style={dashed, draw=red}},
        legend columns=4,
        legend style={
            at={(0.5, 1.26)},
            anchor=north,
            draw=none,
            font=\footnotesize,
        },
        visualization depends on=y \as \rawy,
        visualization depends on=x \as \rawx,
        legend image code/.code={\draw [#1] (0cm,-0.1cm) rectangle (0.1cm,0.1cm);},
        nodes near coords style={
            font=\tiny,
            /pgf/number format/fixed,
            /pgf/number format/fixed zerofill,
            /pgf/number format/precision=2,
            anchor=south,
            %inner xsep=0pt,
            %shift={(axis direction cs:0,-\rawy)}},
            %shift={(axis direction cs:0,0)}
            xshift=3.65\pgfkeysvalueof{/pgf/bar width},
            yshift=-1.65\pgfkeysvalueof{/pgf/bar width},
        },
    ]
%    \addplot[fill=\Clg] table [x=index, y=Ligra-S, col sep=space] {plot/BC.csv};
%    \addlegendentry{Ligra}
    \addplot[fill=\ClgC] table [x=index, y=Ligra-C, col sep=space] {plot/BC.csv};
    \addlegendentry{Ligra ($t=1$)}
%    \addplot[fill=\Cgm] table [x=index, y=Gemini-S, col sep=space] {plot/BC.csv};
%    \addlegendentry{Gemini}
    \addplot[fill=\CgmC] table [x=index, y=Gemini-C, col sep=space] {plot/BC.csv};
    \addlegendentry{Gemini ($t=1$)}
    \addplot[fill=\Cgt] table [x=index, y=GraphIt-S, col sep=space] {plot/BC.csv};
    \addlegendentry{GraphIt ($t=10$)}
%    \addplot[fill=\CgtC] table [x=index, y=GraphIt-C, col sep=space] {plot/BC.csv};
%    \addlegendentry{GraphIt ($t=1$)}
    \addplot[fill=black, nodes near coords] table [x=index, y=bufGraph, col sep=space] {plot/BC.csv};
    \addlegendentry{\bufGraph{}}
    \node[fill=white, xshift=-2pt, rotate=0, font=\bfseries\tiny] at (axis cs:1, 2.65) {5.48};
    \node[fill=white, xshift=-2pt, rotate=0, font=\bfseries\tiny] at (axis cs:2, 2.65) {14.30};
    \end{axis}
\end{tikzpicture}
        \captionsetup{aboveskip=-2pt}
        %\captionsetup{belowskip=-10pt}
        \caption{BC}
        \label{fig:BC}
    \end{subfigure}
    \begin{subfigure}[t]{0.49\textwidth}
        \centering
        \begin{tikzpicture}
    \begin{axis}[
        y tick label style={
            /pgf/number format/fixed,
            /pgf/number format/fixed zerofill,
            /pgf/number format/precision=1,
        },
yticklabel style = {font=\footnotesize},
xticklabel style = {font=\footnotesize},
ylabel style={font=\footnotesize},
xlabel style={font=\footnotesize},
        ytick distance = 0.2,
        enlarge x limits=0.2,
        ymax = 1.1,
        ymin = 0,
        width=84mm,
        height=36mm,
        ylabel=Normalized execution time,
        %xlabel=data sets,
        ybar=0pt,
        bar width=3pt,
        xtick = {1, 2, 3, 4, 5},
        xticklabels = {\texttt{Or}, \texttt{Wk}, \texttt{Lj}, \texttt{Pt}, \texttt{Tw}},
        extra y ticks = 1,
        extra y tick labels={},
        extra y tick style={grid=major,major grid style={dashed, draw=red}},
        legend columns=4,
        legend style={
            at={(0.5, 1.26)},
            anchor=north,
            draw=none,
            font=\footnotesize,
        },
        visualization depends on=y \as \rawy,
        visualization depends on=x \as \rawx,
        legend image code/.code={\draw [#1] (0cm,-0.1cm) rectangle (0.1cm,0.1cm);},
        nodes near coords style={
            font=\tiny,
            /pgf/number format/fixed,
            /pgf/number format/fixed zerofill,
            /pgf/number format/precision=2,
            anchor=south,
            %inner xsep=0pt,
            %shift={(axis direction cs:0,-\rawy)}},
            %shift={(axis direction cs:0,0)}
            xshift=3.8\pgfkeysvalueof{/pgf/bar width},
        },
    ]
    %\addplot[fill=\Clg] table [x=index, y=Ligra-S, col sep=space] {plot/NCP.csv};
    %\addlegendentry{Ligra}
    \addplot[fill=\ClgC] table [x=index, y=Ligra-C, col sep=space] {plot/NCP.csv};
    \addlegendentry{Ligra ($t=1$)}
    %\addplot[fill=\Cgm] table [x=index, y=Gemini-S, col sep=space] {plot/NCP.csv};
    %\addlegendentry{Gemini}
    \addplot[fill=\CgmC] table [x=index, y=Gemini-C, col sep=space] {plot/NCP.csv};
    \addlegendentry{Gemini ($t=1$)}
    %\addplot[fill=\Cgt] table [x=index, y=GraphIt-S, col sep=space] {plot/NCP.csv};
    %\addlegendentry{GraphIt}
    \addplot[fill=\CgtC] table [x=index, y=GraphIt-C, col sep=space] {plot/NCP.csv};
    \addlegendentry{GraphIt ($t=1$)}
    \addplot[fill=black, nodes near coords] table [x=index, y=bufGraph, col sep=space] {plot/NCP.csv};
    \addlegendentry{\bufGraph{}}
    \end{axis}
\end{tikzpicture}
        \captionsetup{aboveskip=-2pt}
        %\captionsetup{belowskip=-10pt}
        \caption{NCP}
        \label{fig:NCP}
    \end{subfigure}
    \begin{subfigure}[t]{0.49\textwidth}
        \centering
        \begin{tikzpicture}
    \begin{axis}[
        y tick label style={
            /pgf/number format/fixed,
            /pgf/number format/fixed zerofill,
            /pgf/number format/precision=1,
        },
yticklabel style = {font=\footnotesize},
xticklabel style = {font=\footnotesize},
ylabel style={font=\footnotesize},
xlabel style={font=\footnotesize},
        ytick distance = 0.5,
        enlarge x limits=0.2,
        ymax = 3,
        ymin = 0,
        width=84mm,
        height=36mm,
        ylabel=Normalized execution time,
        %xlabel=data sets,
        ybar=0pt,
        bar width=3pt,
        xtick = {1, 2, 3, 4, 5},
        xticklabels = {\texttt{Ca}, \texttt{Us}, \texttt{Eu}, \texttt{Wk}, \texttt{Pt}},
        extra y ticks = 1,
        extra y tick labels={},
        extra y tick style={grid=major,major grid style={dashed, draw=red}},
        legend columns=4,
        legend style={
            at={(0.5, 1.26)},
            anchor=north,
            draw=none,
            font=\footnotesize,
        },
        visualization depends on=y \as \rawy,
        visualization depends on=x \as \rawx,
        legend image code/.code={\draw [#1] (0cm,-0.1cm) rectangle (0.1cm,0.1cm);},
        nodes near coords style={
            font=\tiny,
            /pgf/number format/fixed,
            /pgf/number format/fixed zerofill,
            /pgf/number format/precision=2,
            anchor=south,
            %inner xsep=0pt,
            %shift={(axis direction cs:0,-\rawy)}},
            %shift={(axis direction cs:0,0)}
            xshift=3.8\pgfkeysvalueof{/pgf/bar width},
        },
    ]
%    \addplot[fill=\Clg] table [x=index, y=Ligra-S, col sep=space] {plot/LL.csv};
%    \addlegendentry{Ligra}
    \addplot[fill=\ClgC] table [x=index, y=Ligra-C, col sep=space] {plot/LL.csv};
    \addlegendentry{Ligra ($t=1$)}
%    \addplot[fill=\Cgm] table [x=index, y=Gemini-S, col sep=space] {plot/LL.csv};
%    \addlegendentry{Gemini}
    \addplot[fill=\CgmC] table [x=index, y=Gemini-C, col sep=space] {plot/LL.csv};
    \addlegendentry{Gemini ($t=1$)}
    \addplot[fill=\Cgt] table [x=index, y=GraphIt-S, col sep=space] {plot/LL.csv};
    \addlegendentry{GraphIt ($t=10$)}
%    \addplot[fill=\CgtC] table [x=index, y=GraphIt-C, col sep=space] {plot/LL.csv};
%    \addlegendentry{GraphIt-C}
    \addplot[fill=black, nodes near coords] table [x=index, y=bufGraph, col sep=space] {plot/LL.csv};
    \addlegendentry{\bufGraph{}}
    \node[fill=white, xshift=-2pt, rotate=0, font=\bfseries\tiny] at (axis cs:1, 2.65) {5.86};
    \node[fill=white, xshift=-2pt, rotate=0, font=\bfseries\tiny] at (axis cs:2, 2.65) {26.07};
    \end{axis}
\end{tikzpicture}
        \captionsetup{aboveskip=-2pt}
        %\captionsetup{belowskip=-10pt}
        \caption{LL}
        \label{fig:LL}
    \end{subfigure}
    \captionsetup{aboveskip=3pt}
    %\captionsetup{aboveskip=0pt}
	%\captionsetup{belowskip=-2pt}
    \caption{Overall execution time for the three applications with different implementations. \SR{We carefully tune all the systems and only report the best performance.}}
    \label{fig:exp-overall}
\end{figure}
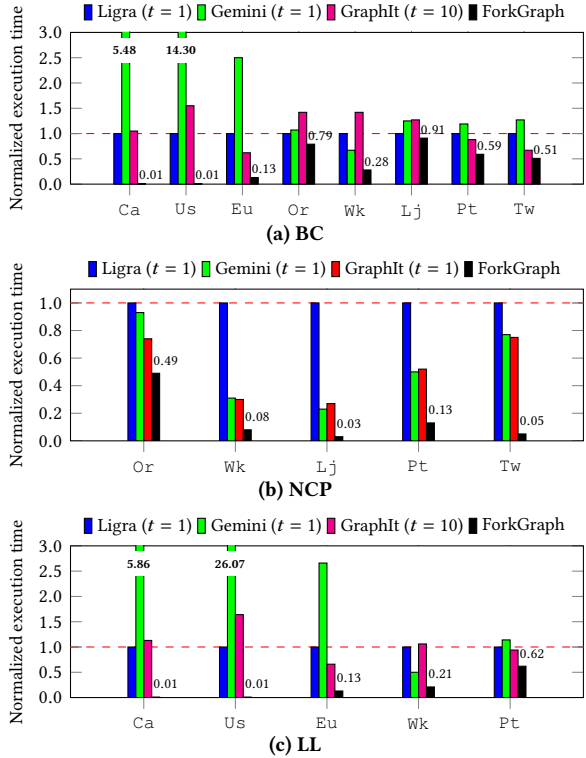

\noindent\fbox{\parbox{\linewidth}{\textbf{Finding (1):} \bufGraph{} significantly outperforms Ligra, Gemini, and GraphIt in different execution schemes by $32\times$, $307\times$, and $38\times$ speedups on average, respectively.}}

Figure~\ref{fig:exp-overall} shows the performance of the applications with different implementations. For Ligra, Gemini, and GraphIt, we evaluate different threading configurations ($t$), and only show the best results for brevity. As the execution time of different test cases varies significantly, we present the normalized execution time to the performance of Ligra ($t=1$). The normalized performance of \bufGraph{} is annotated on the plots.

In Figures~\ref{fig:BC} and \ref{fig:LL}, since Ligra ($t=1$), Gemini ($t=1$), and GraphIt ($t=10$) outperform Ligra ($t=10$), Gemini ($t=10$), and GraphIt ($t=1$), respectively, for almost all the cases, we omit the results for brevity. Similarly, we omit the results of Ligra ($t=10$), Gemini ($t=10$), and GraphIt ($t=10$) in Figure~\ref{fig:NCP}.

Overall, \bufGraph{} significantly outperforms the other three GPSs in different schemes on all the tested applications with two orders of magnitude speedups on average. Besides, we show later that \bufGraph{} reduces the number of LLC misses by more than a factor of $10\times$. We make the following observations in comparison with each of state-of-the-art GPSs.

First, \bufGraph{} shows $51\times$ speedups over Ligra ($t=10$) and $32\times$ over Ligra ($t=1$) on average.  \bufGraph{} accelerates the convergence within partitions with low cache thrashing and adopts the sequential implementation to reduce the total workload.

Second, Gemini's implementations suffer from high synchronization overhead in each iteration because it is designed with the message passing mechanism for a distributed setting. Although all the message-passing functions are disabled in our evaluation, the materialization overhead between consecutive iterations is significant. \bufGraph{} handles the operations of FPP queries using fast implementations of sequential algorithms, which incurs minimal overhead in synchronization and atomic operations. As a result, \bufGraph{} delivers three orders of magnitude speedups over Gemini, especially on road graphs with high diameters.

Third, although GraphIt optimizes a single query's cache usage, it shows higher contention with more threads enabled. Except for solving PPR, GraphIt shows a slowdown when $t=1$. \SR{The reason is that the graph-traversal based queries in BC and LL benefit the cache optimization provided in GraphIt, while the high LLC misses due to uncoordinated memory accesses is too high to be covered by the performance gain when leveraging the inter-query parallelism in GraphIt ($t=1$).} Compared with GraphIt, \bufGraph{} aims to optimize the performance in the inter-query parallelism setting and achieves up to $197\times$ speedups over the best of GraphIt under different schemes. \bufGraph{} slightly outperforms GraphIt ($t=10$) on social networks by $1.3\times$ on solving BC because GraphIt ($t=10$) generates more efficient direction-optimized traversals by searching through a much larger space of optimizations. On data sets that do not rely on direction optimization, \bufGraph{} achieves more than $100\times$ speedups over both GraphIt ($t=10$) and GraphIt ($t=1$).

\begin{table}[t]
    \captionsetup{aboveskip=0pt}
    %\captionsetup{belowskip=0pt}
    \caption{Execution time and \SR{memory consumption} of NCP on different data sets using the four systems.}
    \label{tab:time}
    %\begin{table}[t]
%\centering
%\caption{Graph inputs}
%\label{tab:dataset}
%\begin{tabular}{r|r|r}
%Data sets & Vertices & Num. Edges\\ \hline
%CAL~\cite{demetrescu20089th} & 1.9M   & 4.7M      \\ \hline
%USA~\cite{demetrescu20089th} & 23.9M    & 58.3M       \\ \hline
%LiveJ~\cite{snapnets} & 4.8M    & 68.5M     \\ \hline
%Wiki~\cite{davis2011university}  & 16.8M & 99.8M \\ \hline
%UK-2002~\cite{BCSU3}  & 18.4M & 261.8M \\ \hline
%Twitter~\cite{Kwak10www}  & 41.7M    & 1.5B      \\ \hline
%Friendster~\cite{snapnets}  & 68.3M    & 2.6B      \\ \hline
%\end{tabular}
%\end{table}

\small
\centering
%\resizebox{0.46\textwidth}{!}{
    \begin{tabular}{|lrrrrr|}
        \multicolumn{6}{c}{\textbf{A. Execution time (minutes)}}\\
        \Xhline{2\arrayrulewidth}
         & \texttt{Or}  & \texttt{Wk}    & \texttt{Lj}     & \texttt{Pt}     & \texttt{Tw}\\ \hline
        Ligra ($t=10$)  & 0.7 & 10.1 & 23.5 & 28.9 & 692.0 \\
        Ligra ($t=1$)   & 0.7 & 8.1  & 19.8 & 17.8 & 243.3 \\\hline
        Gemini ($t=10$) & 0.7 & 4.2  & 7.2  & 13.2 & 230.8 \\
        Gemini ($t=1$)  & 0.7 & 2.5  & 4.5  & 8.9  & 187.4 \\\hline
        GraphIt ($t=10$)& 0.6 & 3.5  & 7.1  & 11.2 & 198.8 \\
        GraphIt ($t=1$) & 0.5 & 2.4  & 5.4  & 9.3  & 181.9 \\\hline
        \bufGraph{}     & 0.4 & 0.6  & 0.7  & 2.3  & 11.7\\
        \Xhline{2\arrayrulewidth}
        \multicolumn{6}{c}{\SR{\textbf{B. Memory consumption (GB)}}}\\
        \Xhline{2\arrayrulewidth}
        & \texttt{Or}  & \texttt{Wk}    & \texttt{Lj}     & \texttt{Pt}     & \texttt{Tw}\\ \hline
        Ligra & 7.9 & 17.2 & 39.7 & 70.3 & 148.0 \\\hline
        Gemini & 16.1 & 19.9 & 35.0 & 103.2 & 130.8 \\\hline
        GraphIt & 17.1 & 20.0 & 37.3 & 103.9 & 149.8 \\\hline
        \bufGraph{} & 12.7 & 23.8 & 39.1 & 98.6 & 152.1 \\
        \Xhline{2\arrayrulewidth}
    \end{tabular}
%}
\end{table}

As an example that details the actual execution time, Table~\ref{tab:time}A shows the execution times of the four systems on solving NCP. We can observe that it takes Ligra ($t=10$) 11.5 hours to process the PPRs on the \texttt{Tw} graph, while \bufGraph{} only needs 12 minutes. Thus, \bufGraph{} is more practical for many graph applications. \SR{Table~\ref{tab:time}B shows the memory usages of the four systems. Basically, most of the memory is spent on storing the execution results of FPP queries, and ForkGraph consumes $5-19\%$ more memory than other GPSs, mainly caused by buffers.}

\subsection{Cache Efficiency and Work Efficiency} \label{subsec:eval:cache}
\noindent\fbox{\parbox{\linewidth}{\textbf{Finding (2):} \bufGraph{} shows up to a factor of $100\times$ reduction of the number of LLC misses. First, the buffered execution is cache-efficient and it reduces the LLC misses of \bufGraph{} even with the same amount of work \CR{as other GPSs}. Second, the work efficient design of FPP queries processing further reduces the amount of total LLC accesses.}}

Figure~\ref{fig:llc-work} shows the profiling of cache performance and the number of edges processed of different GPSs. We present the amount of work as the number of edges processed during processing. For brevity, we show the profiling of LL and NCP applications on two representative graphs for each. We observe similar findings on other graphs. With the buffered execution model and work-efficient optimizations, \bufGraph{} significantly reduces the total LLC misses and work compared to others. Figure~\ref{fig:exp-llc} shows that \bufGraph{} reduces the number of LLC misses by up to a factor of $100\times$, compared to other GPSs' execution with $t=1$, and by up to a factor of $60\times$ over others' execution with $t=10$.

\SR{We also count the number of edges processed in executing sequential algorithms on the same data sets. We find that \bufGraph{} only processes $10.4-16.7\times$ more edges than the sequential algorithm (Dijkstra's algorithm) on BC and LL and $5.2-9.4\times$ more than the sequential algorithm on NCP.}

The significant reduction of LLC misses mainly comes from two optimizations: 1) the LLC-sized partition helps to guarantee that operations are limited within LLC during intra-partition processing and 2) the total reduction of workloads as \bufGraph{} \CR{can be} theoretically proved to be work-efficient. As shown in Figure~\ref{fig:exp-work}, \bufGraph{} significantly reduces the number of edges processed on road networks due to work-efficient yielding and scheduling, which is also the reason for significant speedups over the other GPSs. Although \bufGraph{} shows a similar amount of work as Gemini and GraphIt on \texttt{Lj} and \texttt{Tw}, it only incurs fewer than $1/10$ of the LLC misses as others, delivering near an order of magnitude speedup over them.

\begin{figure}[t]
	\centering
	\begin{subfigure}{0.5\textwidth}
		\begin{tikzpicture}
    \begin{axis}[
        legend cell align={left},
        % y tick label style={
        %     /pgf/number format/fixed,
        %     /pgf/number format/fixed zerofill,
        %     /pgf/number format/precision=1,
        % },
yticklabel style = {font=\footnotesize},
xticklabel style = {font=\footnotesize},
ylabel style={font=\footnotesize},
xlabel style={font=\footnotesize},
        %ytick distance = 0.2,
        enlarge x limits=0.1,
        width=84mm,
        height=36mm,
        ymode=log,
        %log y ticks with fixed point,
        ylabel=\#LLC misses,
        ybar=0pt,
        bar width=3pt,
        xtick = {1, 2, 3, 4},
        xticklabels = {LL on \texttt{Ca}, LL on \texttt{Us}, NCP on \texttt{Lj}, NCP on \texttt{Tw}},
        extra y ticks = 1,
        extra y tick labels={},
        extra y tick style={grid=major,major grid style={dashed, draw=red}},
        legend columns=4,
        legend style={
            at={(0.5, 1.02)},
            anchor=south,
            draw=none,
            font=\scriptsize,
        },
        legend image code/.code={\draw [#1] (0cm,-0.1cm) rectangle (0.1cm,0.1cm);},
    ]
    \addplot[fill=\ClgC] table [x=index, y=Ligra-S, col sep=space] {plot/LLC.csv};
    \addlegendentry{Ligra ($t=10$)}
    \addplot[fill=\ClgC!40] table [x=index, y=Ligra-C, col sep=space] {plot/LLC.csv};
    \addlegendentry{Ligra ($t=1$)}
    \addplot[fill=\CgmC] table [x=index, y=Gemini-S, col sep=space] {plot/LLC.csv};
    \addlegendentry{Gemini ($t=10$)}
    \addplot[fill=\CgmC!40] table [x=index, y=Gemini-C, col sep=space] {plot/LLC.csv};
    \addlegendentry{Gemini ($t=1$)}
    \addplot[fill=\CgtC] table [x=index, y=GraphIt-S, col sep=space] {plot/LLC.csv};
    \addlegendentry{GraphIt ($t=10$)}
    \addplot[fill=\CgtC!40] table [x=index, y=GraphIt-C, col sep=space] {plot/LLC.csv};
    \addlegendentry{GraphIt ($t=1$)}
    \addplot[fill=black] table [x=index, y=bufGraph, col sep=space] {plot/LLC.csv};
    \addlegendentry{\bufGraph{}}
    \addplot[fill=orange] table [x=index, y=Sequential, col sep=space] {plot/LLC.csv};
    \addlegendentry{Sequential}
    \end{axis}
\end{tikzpicture}
		\captionsetup{aboveskip=0pt}
		\caption{Number of LLC misses.}
		\label{fig:exp-llc}
	\end{subfigure}
	\begin{subfigure}{0.5\textwidth}
		\begin{tikzpicture}
    \begin{axis}[
        legend cell align={left},
        % y tick label style={
        %     /pgf/number format/fixed,
        %     /pgf/number format/fixed zerofill,
        %     /pgf/number format/precision=1,
        % },
yticklabel style = {font=\footnotesize},
xticklabel style = {font=\footnotesize},
ylabel style={font=\footnotesize},
xlabel style={font=\footnotesize},
        %ytick distance = 0.2,
        enlarge x limits=0.1,
        width=84mm,
        height=36mm,
        ymode=log,
        %log y ticks with fixed point,
        ylabel=\#edges processed,
        ybar=0pt,
        bar width=3pt,
        xtick = {1, 2, 3, 4},
        xticklabels = {LL on \texttt{Ca}, LL on \texttt{Us}, NCP on \texttt{Lj}, NCP on \texttt{Tw}},
        extra y ticks = 1,
        extra y tick labels={},
        extra y tick style={grid=major,major grid style={dashed, draw=red}},
        legend columns=4,
        legend style={
            at={(0.5, 1.52)},
            anchor=north,
            draw=none,
            font=\scriptsize,
        },
        legend image code/.code={\draw [#1] (0cm,-0.1cm) rectangle (0.15cm,0.2cm);},
    ]
    \addplot[fill=\ClgC] table [x=index, y=Ligra-S, col sep=space] {plot/Work.csv};
    %\addlegendentry{Ligra}
    \addplot[fill=\ClgC!40] table [x=index, y=Ligra-C, col sep=space] {plot/Work.csv};
    %\addlegendentry{Ligra-C}
    \addplot[fill=\CgmC] table [x=index, y=Gemini-S, col sep=space] {plot/Work.csv};
    %\addlegendentry{Gemini}
    \addplot[fill=\CgmC!40] table [x=index, y=Gemini-C, col sep=space] {plot/Work.csv};
    %\addlegendentry{Gemini-C}
    \addplot[fill=\CgtC] table [x=index, y=GraphIt-S, col sep=space] {plot/Work.csv};
    %\addlegendentry{GraphIt}
    \addplot[fill=\CgtC!40] table [x=index, y=GraphIt-C, col sep=space] {plot/Work.csv};
    %\addlegendentry{GraphIt-C}
    \addplot[fill=black] table [x=index, y=bufGraph, col sep=space] {plot/Work.csv};
    %\addlegendentry{\bufGraph{}}
    \addplot[fill=orange] table [x=index, y=Sequential, col sep=space] {plot/Work.csv};
    %\addlegendentry{\bufGraph{}}
    \end{axis}
\end{tikzpicture}
		\captionsetup{aboveskip=0pt}
		\caption{Number of edges processed.}
		\label{fig:exp-work}
	\end{subfigure}
	\captionsetup{aboveskip=0pt}
	%\captionsetup{belowskip=-12pt}
	\caption{Profiling results of the number of LLC misses and the number of edges processed per FPP query on four GPSs, solving two applications on four graphs.}
	\label{fig:llc-work}
\end{figure}
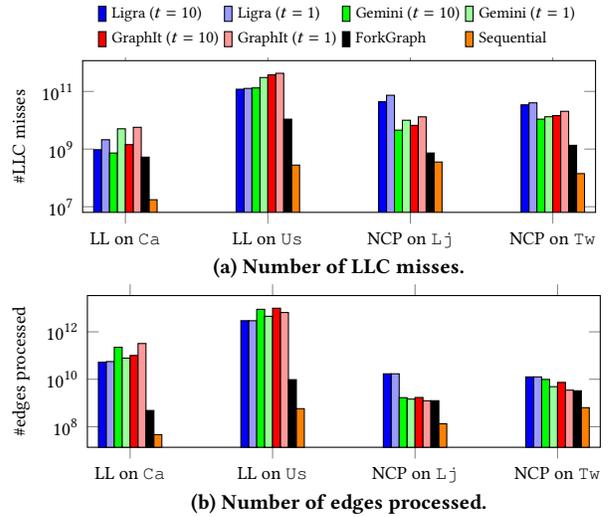

%%%%%%%%%%%%%%%%%%%%%%%%%%%%%%%%%
%%  Optimization cumulatively
%%%%%%%%%%%%%%%%%%%%%%%%%%%%%%%%%
\subsection{Effects of Individual Techniques}\label{subsec:eval:individual}
\begin{figure}[t]
	\centering
	\begin{tikzpicture}
    \begin{axis}[
        y tick label style={
            /pgf/number format/fixed,
            /pgf/number format/fixed zerofill,
            /pgf/number format/precision=0,
        },
yticklabel style = {font=\footnotesize},
xticklabel style = {font=\footnotesize},
ylabel style={font=\footnotesize},
xlabel style={font=\footnotesize},
        ytick distance = 50,
        enlarge x limits=0.15,
        width=85mm,
        height=36mm,
        ymin=0,
        %ymode=log,
        %log y ticks with fixed point,
        ylabel=Speedup,
        ybar=0pt,
        bar width=6pt,
        xtick = {1, 2, 3, 4},
        xticklabels = {LL on \texttt{Ca}, LL on \texttt{Us}, NCP on \texttt{Lj}, NCP on \texttt{Tw}},
        % extra y ticks = 1,
        % extra y tick labels={},
        % extra y tick style={grid=major,major grid style={dashed, draw=red}},
        legend columns=-1,
        legend style={
            at={(0.5, 1.27)},
            anchor=north,
            draw=none,
            font=\footnotesize,
        },
        legend image code/.code={\draw [#1] (0cm,-0.1cm) rectangle (0.1cm,0.1cm);},
                nodes near coords style={
            font=\tiny,
            /pgf/number format/fixed,
            /pgf/number format/fixed zerofill,
            /pgf/number format/precision=2,
            anchor=south,
            %inner xsep=0pt,
            %shift={(axis direction cs:0,-\rawy)}},
            %shift={(axis direction cs:0,0)}
            %xshift=3.65\pgfkeysvalueof{/pgf/bar width},
            %yshift=-1.65\pgfkeysvalueof{/pgf/bar width},
        },
    ]
    % \addplot[fill=\ClgC] table [x=index, y=none, col sep=space] {plot/acc.csv};
    % \addlegendentry{Ligra}
    \addplot[fill=\CgmC] table [x=index, y=buffer, col sep=space] {plot/acc.csv};
    \addlegendentry{+buffer}
    \addplot[fill=\CgtC] table [x=index, y=gathering, col sep=space] {plot/acc.csv};
    \addlegendentry{+consolidation}
    \addplot[fill=yellow] table [x=index, y=ordering, col sep=space] {plot/acc.csv};
    \addlegendentry{+priority scheduling}
    \addplot[fill=black] table [x=index, y=yield, col sep=space] {plot/acc.csv};
    \addlegendentry{+yielding}
    \node[fill=none, xshift=-12pt, rotate=0, font=\bfseries\tiny] at (axis cs:3, 10) {0.48};
    \node[fill=none, xshift=-12pt, rotate=0, font=\bfseries\tiny] at (axis cs:4, 10) {0.14};
    \end{axis}
\end{tikzpicture}
	\captionsetup{aboveskip=0pt}
	%\captionsetup{belowskip=0pt}
	\caption{Speedups achieved by applying different optimizations cumulatively to the Ligra baseline.}
	\label{fig:acc}
\end{figure}
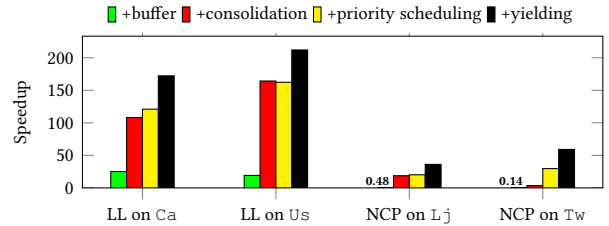

\noindent\fbox{\parbox{\linewidth}{\textbf{Finding (3):} The proposed techniques accumulatively improve the performance of \bufGraph{}.}}

We evaluate the performance impacts of major design rationales in \bufGraph{}, including buffered execution, query-centric operation consolidation, priority-based scheduling, and yielding. We cumulatively enable each of those optimizations one by one on the baseline Ligra ($t=10$) to study the advantages of individual design decisions. For brevity, we show LL and NCP on four graphs again.

Figure~\ref{fig:acc} shows the performance improvement. We have the following observations. First, we enable the \bufModel{} and sequential execution of FPP queries, denoted as +buffer. +buffer achieves $25\times$ speedups on solving LL on road networks because it benefits from the good locality. Also, we observe that applying the \bufModel{} brings negative performance improvements in solving NCP. The reason is that we only finish the processing of a PPR query in the partition when it converges, which brings extra workload when dealing with an excessive number of revisits. Second, we enable the atomic-free processing by query-centric operation consolidation, denoted as +consolidation. It begins to show significant speedups by both reducing the overhead of atomic operations and the memory overhead caused by expanding superfluous operations to neighbor partitions. Next, we enable the optimizations of priority-based scheduling and yielding, respectively. These two optimizations holistically improve the performance over +consolidation by further $1.3-16.2\times$. The yielding optimization generally provides more significant speedups than others as it cuts off the work directly during the processing, while the priority-based scheduling reduces workloads indirectly.

\vspace{1mm}\noindent
\textbf{Impacts of parameter tuning in inter-partition scheduling.} We further evaluate the performance impacts of the priority-based scheduling and yielding. We only show the results for BC with 100 SSSPs on \texttt{Us} graph for brevity.

Table~\ref{tab:priority}A shows the impact of different priority functors. We make the following observations. First, it is inefficient to schedule the partition with most operations to process (``Max \#operations''), even though it is more cache efficient intuitively. As a result, it is even slower than the baseline by picking an arbitrary non-empty buffer to process in each step (``Random''). \SR{Second, the performance of the default FIFO scheduling is slightly better than a random scheduling.} Third, compared to other priority functors, adopting the ``Shortest'' priority functor from the corresponding sequential algorithms delivers several times speedups over others.

% It implies that the optimization proposed in sequential algorithm could still perform well on a granularity of partitions. Looking forward, more reproductions to theoretically efficient algorithm can be potentially beneficial to the efficient designs in GPSs.

Table~\ref{tab:priority}B shows the performance of \bufGraph{} using different yielding thresholds based on the number of edges processed. We define $\mu$ to be the number of edges in the partition divided by the total number of queries\CR{, which is the theoretical threshold
\ifdefined\complete{}
 (see the proof in \appref{app:work}).
\else
 (see the proof in the appendix of the complete version~\cite{forkgraph2021}).
\fi}
We change the threshold value of the heuristics at the basis of $\mu$. We can observe that the \CR{execution of applying }threshold $\mu$ results in an execution time near the fastest but not necessarily the fastest. It is because that the threshold is obtained based on the theoretical upper bound of the number of revisits; however, the number of revisits is far below the bound in practice, which makes a larger threshold perform well. \SR{As there can be thousands of queries in our experiments for NCP applications on large graphs, we use a larger threshold, $100\mu$, for these cases.}

Table~\ref{tab:priority}C shows the impact of yielding heuristics based on the value updated. We adopt the $\Delta=50,000$ used in ~\cite{zhang2020optimizing} for the same data set \texttt{Us} and also test the execution instances with the threshold varied. We present the execution times of \bufGraph{} \CR{with the threshold setting varied from} $0.25\Delta$ to $4\Delta$ for brevity. We have the observations as follows. First, when the threshold is large, \bufGraph{} spends more time as there are more redundant operations abandoned in each partition. Second, when the threshold is small, \bufGraph{} aggressively yields the processing in a partition and results in a high number of revisits. We choose the thresholds directly adopted from the $\Delta$-stepping algorithm~\cite{meyer2003delta, zhang2020optimizing} in our experiments.

\begin{table}[t]
	\captionsetup{aboveskip=0pt}
	%\captionsetup{belowskip=0pt}
	\caption{Performance of \bufGraph{} under different priority-based scheduling \CR{methods} and yielding parameters, solving BC on \texttt{Us}.}
	\label{tab:priority}
	%\begin{table}[t]
%\centering
%\caption{Graph inputs}
%\label{tab:dataset}
%\begin{tabular}{r|r|r}
%Data sets & Vertices & Num. Edges\\ \hline
%CAL~\cite{demetrescu20089th} & 1.9M   & 4.7M      \\ \hline
%USA~\cite{demetrescu20089th} & 23.9M    & 58.3M       \\ \hline
%LiveJ~\cite{snapnets} & 4.8M    & 68.5M     \\ \hline
%Wiki~\cite{davis2011university}  & 16.8M & 99.8M \\ \hline
%UK-2002~\cite{BCSU3}  & 18.4M & 261.8M \\ \hline
%Twitter~\cite{Kwak10www}  & 41.7M    & 1.5B      \\ \hline
%Friendster~\cite{snapnets}  & 68.3M    & 2.6B      \\ \hline
%\end{tabular}
%\end{table}

\scriptsize
\centering
\resizebox{0.48\textwidth}{!}{
    \begin{tabular}{|l|c|c|c|c|}
        \multicolumn{5}{c}{\textbf{A. Impacts of priority-based scheduling (yielding enabled).}}\\
        \Xhline{2\arrayrulewidth}
        Priority functor  & Random  & Max \#operations & \SR{FIFO} & \textbf{Shortest} \\ \hline
        Execution time (s) & 504.3 & 749.9 & \SR{491.3} & \textbf{168.8} \\
        \Xhline{2\arrayrulewidth}
    \end{tabular}
}
\vspace{1pt}

\normalsize
\centering
\resizebox{0.48\textwidth}{!}{
    \begin{tabular}{|l|c|c|c|c|c|c|}
        \multicolumn{7}{c}{\textbf{B. Impacts of yielding heuristic 1: on the number of edges processed}}\\
        \multicolumn{7}{c}{\textbf{(priority-based scheduling enabled).}} \\
        \Xhline{2\arrayrulewidth}
        Threshold   & $0.25\mu$ & $0.5\mu$ & $\mu$ & $\bm{2\mu}$ & $4\mu$ & \SR{No Yielding}\\  \hline
        Execution time (s)   & 450.9 & 412.4 & 325.6 & \textbf{238.6} & 248.3 & \SR{1945.8}\\
        \Xhline{2\arrayrulewidth}
    \end{tabular}
}
\vspace{1pt}

\normalsize
\centering
\resizebox{0.48\textwidth}{!}{
    \begin{tabular}{|l|c|c|c|c|c|c|}
        \multicolumn{7}{c}{\textbf{C. Impacts of yielding heuristic 2: on the operations' values updated}}\\
        \multicolumn{7}{c}{\textbf{(priority-based scheduling enabled).}}\\
        \Xhline{2\arrayrulewidth}
        Threshold   & $0.25\Delta$ & $0.5\Delta$ & $\bm{\Delta}$ & $2\Delta$ & $4\Delta$ & \SR{No Yielding}\\ \hline
        Execution time (s)   & 420.6      & 297.0     & \textbf{168.8}  & 172.1   & 239.8 & \SR{1945.8} \\
        \Xhline{2\arrayrulewidth}
    \end{tabular}
}

\end{table}

%%%%%%%%%%%%%%%%%%%%%%%%%%%%%%%%%
%%  Other results
%%%%%%%%%%%%%%%%%%%%%%%%%%%%%%%%%
\subsection{Other Results} \label{subsec:eval:other}
Due to the space limit, we summarize some experimental results as follows. 
\ifdefined\complete{}
The reader is referred to the details in the appendix.
\else
The reader is referred to the details in the appendix of the complete version~\cite{forkgraph2021}.
\fi

\vspace{1mm}\noindent
\textbf{Memory stall distribution.} As more than $34\%$ of the execution time of other GPSs is spent on memory \CR{stalls}, the time spent in \bufGraph{} is only less than $20\%$ of the total execution time. \bufGraph{}'s design limits the operations to graph partitions in the LLC and thus reduces the percentage of the costly DRAM accesses, which is shown as the low memory stall distribution.

\vspace{1mm}\noindent
\textbf{Scalability.} We evaluate the scalability of ForkGraph in the numbers of threads and queries. \SR{First, \bufGraph{} can achieve $7-8\times$ speedups when scaling up from one to ten cores (with hyper-threading disabled) for most of the graphs. Second, \bufGraph{} shows the capability to remain at a high throughput with processing growing numbers of FPP queries.}

\vspace{1mm}\noindent
\CR{\textbf{Effects of Partition and Cache Size.} We empirically study the effects of graph partitioning methods, partition sizes, and cache sizes. First, \bufGraph{} on METIS partition shows up to $14.1\times$ and $4.2\times$ speedups over a random partition and Gemini's lightweight partition~\cite{zhu2016gemini}, respectively, when executing LL and BC on different graphs. \bufGraph{} on METIS partition shows $1.1-3.6\times$ speedups over other partitioning methods when executing NCP on different web and social networks. Second, our results show that using LLC-size partitions achieves the best performance for most cases. This is because the intra-partition processing can cause heavy cache thrashing if the partition size is larger than the LLC size. Further, if we divide the graph into small partitions, there can be a large number of partitions to schedule, which incurs a high overhead.}

\ifdefined\complete{}
\vspace{3mm}
\else
\fi
\section{Related Work}\label{sec:related}
\textbf{Graph processing on multi-core architectures.}
There has been substantial works on efficient parallel graph systems and frameworks over the past years, including~\cite{gonzalez2012powergraph, low2010graphlab, malewicz2010pregel, zhang2018graphit, zhu2016gemini} among many others. Ligra~\cite{shun2013ligra}, Galois~\cite{nguyen2013lightweight}, as the representative shared-memory graph processing frameworks. These frameworks are designed with abstractions for users to conduct graph computations while leveraging hardware properties such as memory locality and multi-cores efficiently. GraphIt~\cite{zhang2018graphit} is a DSL for graph processing, which generates parallel implementations of graph applications. GraphIt integrates a scheduling language to ease the exploration of the complicated tradeoff space. We refer the reader to \cite{mccune2015thinking, yan2017big} for excellent surveys of this growing literature.

Zhang et al.~\cite{zhang2018graphit} and Lakhotia et al.~\cite{lakhotia2019gpop} propose to improve the cache utilization by breaking the graph into segments that fit in the LLC. In this way, random accesses at each partition are limited in the cache, avoiding costly memory accesses. Similarly, \bufGraph{} also partitions graphs into LLC-size partitions. Unlike their approaches, \bufGraph{} designs a \bufModel{} to specifically optimize processing queries simultaneously to enable the work-efficient and cache-efficient FPP processing.

\vspace{1mm}\noindent
\textbf{Other related works in graph processing systems.}
\SR{Yan et al. propose Blogel~\cite{yan2014blogel}, a block-centric, distributed graph processing framework. Both \bufGraph{} and Bogel consider a partition of a graph as a computing block (instead of a vertex or an edge). The block-centric computing model \CR{proposed by Bogel} helps decrease the number of iterations compared to a vertex-centric algorithm and also \CR{reduces} the number of messages transmitted through the network in the distributed setting. \bufGraph{} is different from Blogel since Blogel is distributed, and \bufGraph{} is in-memory. Besides, Blogel executes one query at a time, while \bufGraph{} focuses on the efficiency of inter-query parallelism among multiple FPP queries.}

\SR{In addition, Zhao et al. propose GraphM~\cite{zhao2019graphm}, a storage system that efficiently handles consolidated, out-of-core concurrent graph queries. GraphM divides the graph into partitions and sets the highest priority to the partition with the most jobs. However, as shown in Table~\ref{tab:priority}, this approach is inefficient for FPP queries because GraphM's scheduling is designed for the out-of-core, BSP model, but not work-efficient for in-memory execution.}

\vspace{1mm}\noindent
\textbf{Processing multiple graph queries simultaneously.}
Zhang et al. \cite{zhang2018cgraph} propose \CR{CGraph, one of the the state-of-the-art disk-based graph processing systems that handle multiple queries simultaneously.} \CR{CGraph} efficiently amortizes the high disk access cost when processing \CR{multiple} queries simultaneously by a subgraph-based scheduling algorithm. Differently, \bufGraph{} targets cache-efficient FPP queries processing in memory. While CGraph processes all queries \CR{in every iterations}, \bufGraph{} \CR{only processes a subset of} FPP queries in the partition-level granularity, leveraging work-efficient sequential executions. Hauck et al.~\cite{hauck2017can} motivate the experimental studies of processing concurrent queries based on the multi-user setups in classic relational enterprise database environments or web-scale environments. They study the inter- and intra- parallelism of handling different types of graph queries by assigning different threads to different instances of Galios~\cite{nguyen2013lightweight}. However, executing multiple instances contains inevitable contentions of the system resources, including memory and threads. The authors provide an in-depth discussion of the limitation of GPSs in handling multiple queries but do not come out with a suitable solution.

MS-BFS \cite{then2014more} and iBFS \cite{liu2016ibfs} are proposed to accelerate multiple BFS queries using multi-core CPUs and GPUs (Graphics Processing Units), respectively. Instead of visiting vertices individually for each BFS, both work leverage joint frontier queue and bitwise operations for multiple BFS queries. However, the techniques specifically serve only BFS queries, losing the generalities.
%    Both iBFS and MS-BFS are motivated by the observation that the likelihood of multiple BFSs have the same vertices in their frontier queues is high, especially on small-world networks~\cite{amaral2000classes}.
%    , combined with a group tuning strategy to maximize the frontier sharing.
%    Although both MS-BFS and iBFS have reduced the number of random memory accesses and amortizes the high cost of cache misses by executing multiple graph queries concurrently, none has exploited the temporal and spatial similarities across different iterations of FPP queries processing.
%    However, the joint frontier queue and approaches

%\textbf{Prior schemes for cache-efficiency graph processing.}

%Faldu et al.~\cite{faldu2020domain} summarize the hardware schemes into three categories. We here refer the reader to their work~\cite{faldu2020domain}, and a series of related studies~\cite{qureshi2007adaptive, jaleel2008adaptive, jaleel2010high, teran2016minimal, jimenez2017multiperspective, faldu2017leeway, jain2016back, vijaykumar2018case} for details of the hardware-oriented approaches.

\vspace{1mm}\noindent
\textbf{Buffering accesses to index structures.}
The access buffering model proposed in this work is inspired by Zhou and Ross's work~\cite{zhou2003buffering, zhou2004buffering} on buffering accesses to tree-structured indexes, e.g., B$^+$-tree \cite{comer1979ubiquitous}, and many other related buffering techniques~\cite{skopal2007construction,shahvarani2016hybrid, graefe2011modern, 10.1145/1366102.1366105}. The buffering techniques are mainly used for avoiding cache thrashing between query accesses by processing buffered lookups at index nodes. Besides, He et al.~\cite{10.1145/1366102.1366105} \CR{develop} a cache-oblivious design on buffering accesses. Those previous works inspire our buffer execution model. However, the previous studies work on a relatively simpler problem on tree accesses, which always go from the root to leaf nodes. In FPPs, the access pattern is more irregular. A query could start randomly from any vertex and expand to different neighbors, and the access can be repeated, unlike the accesses\CR{ on trees}. Therefore, this work develops novel intra- and inter-partition mechanisms to improve the work and cache efficiency.

\vspace{1mm}\noindent
\SR{\textbf{Other related topics from relational databases.}
The cache-aware techniques in \bufGraph{} have been greatly inspired by the substantial studies from relational databases. Particularly, Harizopoulos et al. propose Qpipe~\cite{harizopoulos2005qpipe}, an relational query engine that exploits overlap across concurrent queries at runtime. Qpipe buffers data pages brought by queries and reuses them for other submitted queries. Moreover, systems like Dora~\cite{pandis2010data}, H-Store~\cite{kallman2008h}, and many other works partition the data logically or physically to enable concurrent transactions execution in parallel on partitioned data, which share the similar spirit of \bufGraph{} on LLC-sized graph partitions. The logging solution proposed by Johnson et al.~\cite{johnson2010aether} aggregates requests from threads to reduce the contention among them by making the requests independent from the number of threads. This is similar to the consolidation process of \bufGraph{}.}

\SR{However, the techniques proposed in previous studies cannot be directly applied in this work since these studies are in relational databases and this work focuses on graph processing. Particularly, this work has the following differences. First, the order of operations in Qpipe and other works in relational databases does not affect the work efficiency, while the order of graph query operations affect the work efficiency. Second, compared with the consolidation techniques proposed by Johnson et al.~\cite{johnson2010aether}, \bufGraph{} not only reduces the contention but also leverages the algorithmic properties such as priority and yielding in graph algorithms to reduce redundant computation.}

\ifdefined\complete{}
\vspace{3mm}
\else
\fi
\section{Conclusions}\label{sec:conclusion}
As graph applications emerge, we observe a common and costly \forkP{} (FPP) that launches many independent queries from different source vertices on the same graph. Our profiling studies demonstrate that existing parallel graph systems suffer from severe cache thrashing due to irregular graph structures and many parallel queries in FPP. Thus, we propose \bufGraph{}, a cache-efficient graph processing system for processing FPPs on in-memory graph data. Specifically, \bufGraph{} embraces a cache-efficient buffer execution model to handle operations of many FPP queries. Moreover, we develop effective intra- and inter-partition mechanisms to improve work efficiency.  Our evaluations on real-world graphs show that \bufGraph{} significantly outperforms state-of-the-art graph processing systems (including Ligra, Gemini, and GraphIt) by two orders of magnitude speedups.

%\lu{In this paper, we propose \bufGraph{}, a graph processing system for processing simultaneously graph queries on shared-memory machines to support algorithms with the \forkP{} efficiently.} Specifically, \bufGraph{} applies a \bufModel{} to handle operations to graph data with high reuse of cache content. More importantly, it explores the designs of work efficiency in handling query processing. Our evaluation shows that \bufGraph{} significantly outperforms popular graph processing systems, delivering an order of magnitude performance improvement over them on average. Future work includes extending the system to support more hardware platforms (e.g., GPU and distributed-memory platforms).
%we demonstrate that using existing graph processing systems to support algorithms with \forkP{} incurs high performance penalty caused by cache inefficiency.

\begin{acks}
This project is supported by the grant ``Asian Institute of Digital Finance'' awarded by National Research Foundation, Singapore and administered by the Infocomm Media Development Authority under its Smart Systems Strategic Research Programme in 2020. Any opinions, findings and conclusions or recommendations expressed in this material are those of the author(s) and do not reflect the views of National Research Foundation, Singapore.
\end{acks}

%%
%% The next two lines define the bibliography style to be used, and
%% the bibliography file.
\bibliographystyle{ACM-Reference-Format}
\bibliography{reference}
\balance

\ifdefined\complete{}

\clearpage
\appendix

\section{Work Efficiency Analysis}\label{app:work}
\SR{The efficiency of a parallel program is determined by the total number of operations, or work that it performs. In this section, we analyze the work needed by \bufGraph{} in processing $|Q|$ FPP queries to show that it is work-efficient because it performs the amount of work, to within a constant factor, as the fastest known sequential algorithm~\cite{blelloch1996parallel}. We outline the analysis as follows.}

\SR{First, in Lemma~\ref{l1}, we show that \bufGraph{} is work-efficient when there is only one FPP query running. Second, in Lemma~\ref{l2}, we show that \bufGraph{} is work-efficient in the processing of an arbitrary query $q \in Q$, when there are $|Q|$ FPP queries running simultaneously. Therefore, in Theorem~\ref{l3}, we can establish that FPP queries processing is work-efficient in \bufGraph{}.}

%Since the complexity is algorithm-dependent,
We show the proof sketch of handling SSSPs as an example, and other graph algorithms like BFS and PPR can be proved similarly. Particularly, we prove that \bufGraph{} performs the same amount of work, to within a constant factor, as the fastest known sequential algorithm for SSSP, which is Dijkstra's algorithm using a Fibonacci heap priority queue, i.e., $\bigO(|E| + |V|log|V|)$~\cite{fredman1987fibonacci}.

% as the sequential algorithm, whose time complexity is $\bigO(log|V|\cdot(|V| + |E|))$.
%We also adopt the proof to the same algorithm using a Fibonacci heap priority queue~\cite{fredman1987fibonacci}, omitted due to space limit.%, we omit the latter one, whose proof is much longer. %Interested readers can find details in the full version of this paper.

%\begin{lemma}\label{l0}
%    \bufGraph{} performs a constant amount of work $C$ when processing a single SSSP within an arbitrary partition $P_i$.
%\end{lemma}
%\begin{proof}[Proof of Lemma \ref{l0}]

%\end{proof}

\begin{lemma}\label{l1}
    \bufGraph{} executes a single SSSP with the same amount of work, to within a constant factor as the Dijkstra's algorithm $\bigO(|E| + |V|log|V|)$.
\end{lemma}

\vspace{-3mm}
\begin{proof}[Proof of Lemma \ref{l1}]

\SR{In each processing step, \bufGraph{} selects a partition using the inter-partition scheduling and then conducts the intra-partition processing within the partition. Therefore, we compute the work as the product of the following two parts.}

\vspace{1mm}
\noindent
\SR{\emph{Part 1 - inter-partition scheduling:}} When there is only one query processed, \bufGraph{} finds the shortest path to at least one vertex in each step because the priority-based scheduling always selects the partition with the shortest path to process next. Therefore, to find all the shortest paths from the source vertex to other $|V-1|$ vertices, the total number of scheduling is $\bigO(|V|)$. As the cost of maintaining a priority queue is $\bigO(log|V|)$, the number of inter-partition scheduling is $\bigO(|V|log|V|)$.

\vspace{1mm}
\noindent
\SR{\emph{Part 2 - intra-partition processing:}} Given an arbitrary partition $P_i$, the work of processing the SSSP query in $P_i$ is the same as its time complexity, i.e., $\bigO(|E_{P_i}| + |V_{P_i}|log|V_{P_i}|)$, since we adopt the sequential implementation (without yielding optimization). Based on the fact that $P_i$ is an LLC-sized partition, the number of vertices $|V_{P_i}|$ and the number of edges $|E_{P_i}|$ are smaller than the cache size, considered as constant values. Therefore, we can use a constant value $C$ to denote the work of intra-partition processing. Thus, the total work is computed as $\bigO(C\cdot|V|log|V|)$, which is within a constant factor as the Dijkstra's algorithm.
\end{proof}

\begin{lemma}\label{l2}
    When \bufGraph{} executes $|Q|$ SSSP queries, the amount of work it performs on an arbitrary query $q$, is the same, to within a constant factor, as the Dijkstra's algorithm.
\end{lemma}

\vspace{-3mm}
\begin{proof}[Proof of Lemma \ref{l2}]
\SR{Similar as the proof of Lemma~\ref{l1}, we compute the work as the product of the complexity of inter-partition scheduling and intra-partition processing.}

\SR{\emph{Part 1 - inter-partition scheduling:}} As there are total $|Q|$ queries, the upper bound of the number of partition visits is $\bigO(|Q|\cdot|V|)$, each of which requires scheduling of a partition. Among all the scheduling, up to $\bigO(|V|)$ times of visits are scheduled based on query $q$'s preference, and the rest is based on others. Since \bufGraph{} sets a partition's priority as the best priority among all the queries buffered in the partition, there are at most $|V|$ items in the priority queue, with a maintenance cost of at $\bigO(log|V|)$. Therefore, the number of inter-partition scheduling is $\bigO(|Q|\cdot|V|log|V|)$.

\SR{\emph{Part 2 - intra-partition processing:}} \SR{Although the number of inter-partition scheduling increases by $|Q|$ when there are $|Q|$ queries executing, we can leverage the yielding optimization to directly cut off a query's work in intra-partition processing to amortize the redundancy. We can achieve this by either setting the upper bound of the number of edges to process in $P_i$ as $|E_{P_i}|/|Q|$ or leveraging the $\Delta$-stepping algorithm~\cite{meyer2003delta} to yield the query by tuning the $\Delta$ parameter. Thus, we reduce the work of intra-partition processing by a factor of $|Q|$, which is $C/|Q|$.}

Combining the two parts, \bufGraph{} performs $\bigO((C/|Q|)\cdot|Q|\cdot|V|log|V|) = \bigO(C\cdot|V|log|V|)$ amount of work in processing $q$, which is within a constant factor as the Dijkstra's algorithm. Therefore, the processing of an arbitrary SSSP in \bufGraph{} is work efficient.
\end{proof}

Since \bufGraph{} is efficient in performing on any SSSP, it is easy to show that the processing of $|Q|$ SSSPs in \bufGraph{} performs the amount of work, to within a constant factor, as $\bigO(|Q|\cdot(|E| + |V|log|V|))$. Therefore,
\begin{theorem}\label{l3}
    As the fastest known sequential algorithm, Dijkstra's algorithm, performs $\bigO(|Q|\cdot(|E| + |V|log|V|))$ work on $|Q|$ SSSPs, \bufGraph{} performs the amount of work within a constant factor $C$ of it, which is efficient.\qed
\end{theorem}

\SR{Theoretically, the constant factor $C$ is bounded by the LLC size. We experimentally study that \bufGraph{} only processes $5.2-16.7\times$ more edges than the sequential implementations, while other GPSs could process more than $129\times$ edges.}

\section{More Implementation Details}

\subsection{Operation Consolidation in Buckets}
We include more details of the multi-bucket buffer in this section. Particularly, we allocate $K$ independent buckets for each partition's buffer and we equally divide the queries into $K$ disjoint sets and assign each set of queries to explicitly use one of the buckets in each buffer. We use the illustration in Figure~\ref{fig:app-con} as an example to compare the operation consolidation using either a single buffer or multiple buckets. We let $K=2$ and thus divide the buffer of $P_1$ into two buckets and we can allocate threads to handle different buckets, i.e., thread 1 is assigned to handle $q_1$ and $q_3$ in bucket 1, and thread 2 is assigned to handle $q_2$ and $q_4$ in bucket 2.

%%%%%%%%%%%%%%%%%%%%%%%%%%%%%%%%%
%%  consolidation using buckets
%%%%%%%%%%%%%%%%%%%%%%%%%%%%%%%%%
\begin{figure}[t]
	\centering
	\begin{subfigure}[t]{0.44\textwidth}
		\centering
		\includegraphics[width=0.6\textwidth]{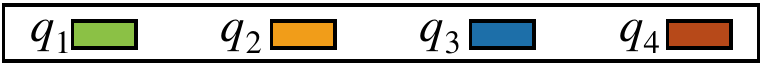}
	\end{subfigure}
	\begin{subfigure}[b]{0.22\textwidth}
		\centering
		\includegraphics[width=0.93\textwidth]{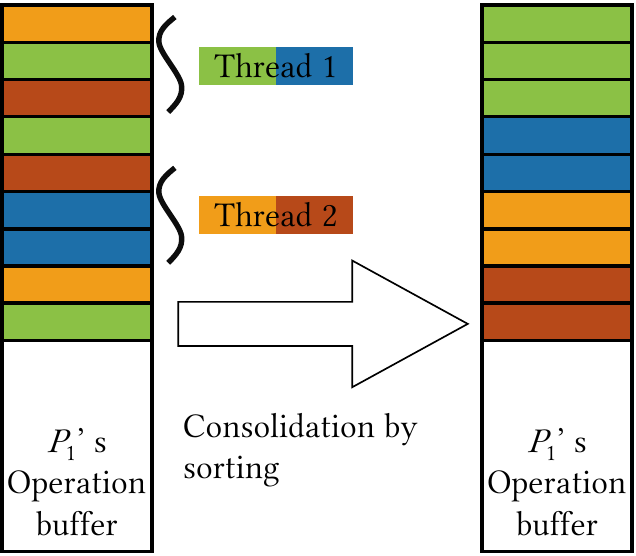}
		\caption{Consolidation in buffer.}
		\label{fig:con-buffer}
	\end{subfigure}\hfill%
	\begin{subfigure}[b]{0.22\textwidth}
		\centering
		\includegraphics[width=0.93\textwidth]{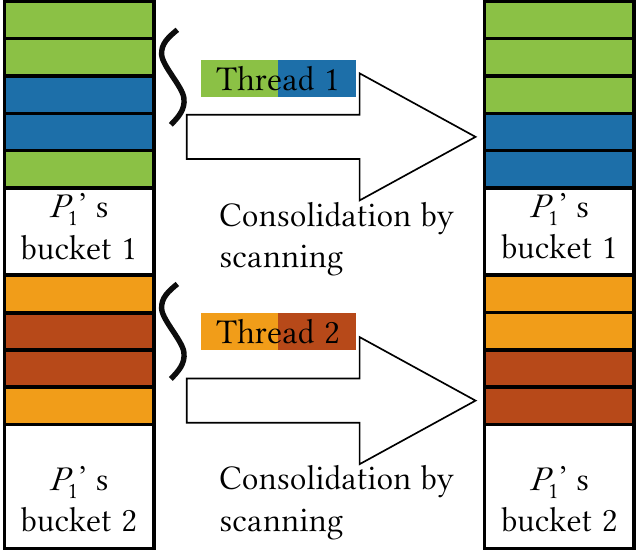}
		\caption{Consolidation in buckets.}
		\label{fig:con-bucket}
	\end{subfigure}
    \captionsetup{aboveskip=0pt}
	%\captionsetup{belowskip=-5pt}
	\caption{Comparison of the operation consolidation using buffer and buckets.}
	\label{fig:app-con}
\end{figure}

\vspace{1mm}
\noindent
\textbf{Complexity analysis}
We include the details of the consolidation process and the complexity analysis. There are two basic methods in operation consolidation in a single buffer. The first method is to sort the buffered operations by their query ids. The second method is to scan the buffer $|Q|$ rounds and select all the operations belonging to one query at one round. We use $R$ to denote the number of buffered operations in a buffer. Thus, the complexities of the two methods are $O(R\log(R))$ and $O(|Q|R)$, respectively. Both methods are costly.

To reduce the cost, we apply the multi-bucket design mentioned above for each buffer. In this way, we can consolidate operations in each bucket independently by either sorting or scanning. Particularly, as we can assume the operations are uniformly distributed among different queries, the complexity of sorting operations in a bucket is $O(\frac{R}{K}\log(\frac{R}{K}))$, and the total cost of $K$ buckets is $O(R\log(\frac{R}{K}))$; the complexity of scanning operations in each bucket is $O(\frac{|Q|}{K}\frac{R}{K})$, and thus the total cost of $K$ buckets is $O(|Q|\frac{R}{K})$. We summarize the complexities of different conditions in Table~\ref{tab:con-complexity}. Essentially, allocating buckets can significantly reduce the cost of consolidation. Note that when we set $K=|Q|$, each bucket only stores operations of a single query, and the cost of consolidation is negligible.

In the implementation, we keep a counter that indicates the number of buffered operations in a bucket. When a thread is ready to send its local operations to the bucket, the thread uses an atomic adder implemented by compare-and-swap (CAS) to increase the counter and thus reserves a range of the space in the bucket. The thread then copies its local operations to the reserved memory space, and there is no conflict issue in this step.
\begin{table}[t]
    \captionsetup{aboveskip=0pt}
	%\captionsetup{belowskip=-10pt}
	\caption{Time complexities of different methods of operation consolidation in buffer or buckets.}
	\centering
	\label{tab:con-complexity}
	\footnotesize
\resizebox{0.495\textwidth}{!}{
		\begin{tabular}{|l|c|c|c|c|}
	    \Xhline{2\arrayrulewidth}
		\textbf{Method} & \textbf{Buffer} & \textbf{Bucket, $\bm{K=2}$} & \textbf{Bucket, $\bm{K=\frac{|Q|}{2}}$} & \textbf{Bucket, $\bm{K=|Q|}$} \\
		\Xhline{2\arrayrulewidth}
		Sort  & $O(R\log(R))$ & $O(R\log(\frac{R}{2}))$      & $O(R\log(\frac{2R}{|Q|})$    & N.A.          \\ \hline
		Scan & $O(|Q|R)$    & $O(\frac{|Q|}{2}R)$ & $O(2R)$                     & N.A.          \\
	    \Xhline{2\arrayrulewidth}
	\end{tabular}
}
\end{table}

\subsection{Independence of Yielding and Priority-based Scheduling}

The yielding and priority-based scheduling are independent. The execution of either one will not affect the effect taken by the other. We give detailed explanations as follows.

On the one hand, the priority-based scheduling does not affect yielding. It only views the priority values of every partition and decides which partition to process next. On the other hand, yielding does not affect the priority-based scheduling as it only decides how much to process queries within one partition. Particularly, different settings of yielding results into the following two cases: 1. Yield the last query in Partition $P_1$ but reschedule $P_1$ back immediately when the priority value of $P_1$ is still the highest among all partitions. 2. Yield the last query in $P_1$ and schedule another Partition $P_x$ to process when the priority value of $P_1$ is no longer the highest. Note that due to the algebraic property (e.g., monotonicity in SSSP) in execution, operations with highest priorities will always be found and also scheduled earlier than others. Similarly, the highest priority will also be used to update other partitions. As a result, $P_x$ will be deterministic based on the priority, which is independent from yielding.

\subsection{Walk-through of Different Scheduling Methods}\label{app:walk-through}

We give the detailed walk-through of the example in Figure~\ref{fig:priority}. We mainly show the operations processed by applying FIFO and Priority-based scheduling. The other two scheduling methods are similar and thus omitted.

In the FIFO scheduling, $P_3$ will be the first partition to process, because it contains the shortest path and thus will be updated first after processing $P_1$. \bufGraph{} processes 2 operations ($\langle q_1, \mathsf{E}, 2\rangle$ and $\langle q_2, \mathsf{E}, 1\rangle$) in $P_3$, and updates $P_2$. Since $P_1$ updates $P_4$ before $P_2$, we then process operations in $P_4$. $P_4$ contains 4 operations but 2 of them can be discarded by consolidation. Therefore, we process $\langle q_1, \mathsf{F}, 3\rangle$ and $\langle q_2, \mathsf{F}, 5\rangle$, and then update $P_2$. $P_2$ now contains 6 operations and those 2 sent from P3 are better than others. Thus, we process operations $\langle q_1, \mathsf{D}, 3\rangle$ and $\langle q_2, \mathsf{D}, 2\rangle$ in $P_2$. We do not need to send $\langle q_1, \mathsf{F}, 4\rangle$ from $P_2$ to $P_4$ since the distance of $\mathsf{F}$ in $P_4$ is 3, which is better. We thus only process $q_2$ in $P_4$. In summary, there are in total 7 operations processed.

In the priority-based scheduling, where we already make the partition with the shortest path the highest priority, we only need to process 6 operations. The process order is $P_3$, $P_2$, and $P_4$. In $P_3$, we first process 2 operations and send 2 to $P_2$. Then we schedule $P_2$, because it contains $\langle q_2, \mathsf{D}, 2\rangle$, the operation with the highest priority. After that, $P2$ sends only 1 operation $\langle q_2, \mathsf{F}, 3\rangle$ to $P_4$, and process 2 operations $\langle q_1, \mathsf{F}, 4\rangle$ and $\langle q_2, \mathsf{F}, 3\rangle$ in $P_4$. Therefore, there are in total 6 operations processed.

\section{More Experimental Results}

\subsection{Time Breakdown of Memory Stalls} Figure~\ref{fig:breakdown} shows the memory stall distribution of the four evaluated GPSs on solving NCP. As more than $34\%$ of the execution time of other GPSs is spent on memory stalls, the time spent in \bufGraph{} is only less than $20\%$ of the total execution time. We can obtain the following implications on FPP queries processing through the evaluations on the number of LLC misses and memory stall distribution. First, although applying inter-query parallelism (i.e., let $t=1$) can achieve better performance than the approach of intra-query parallelism ($t=10$) for most of the cases, a larger number of LLC misses and high memory stalls make the improvement trifling. Second, \bufGraph{}'s design can successfully limit the operations to graph partitions in the LLC and thus reduces the percentage of the costly DRAM accesses, which is shown as the low memory stall distribution in the figure.

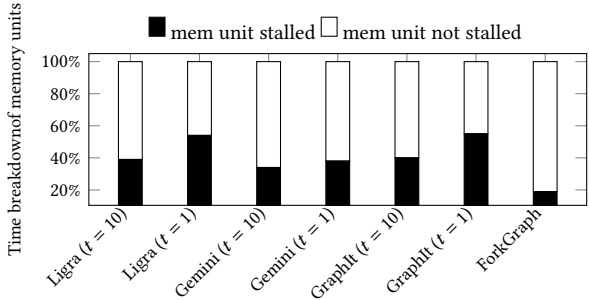
\begin{figure}[h]
	\centering
	\begin{tikzpicture}
\begin{axis}[
yticklabel style = {font=\footnotesize},
xticklabel style = {font=\footnotesize},
ylabel style={font=\footnotesize, align=center},
xlabel style={font=\footnotesize},
    width=82mm,
    height=36mm,
    ybar stacked,
    bar width=9pt,
    ymax=105,
    legend style={
        at={(0.5, 1)},
        anchor=south,
        draw=none,
        fill=none,
        align=left,
        legend columns=2,
        font=\small,
    },
    ylabel=Time breakdown \newline of memory units,
    yticklabel={$\pgfmathprintnumber{\tick}\%$},
    symbolic x coords={Ligra ($t=10$), Ligra ($t=1$), Gemini ($t=10$), Gemini ($t=1$), GraphIt ($t=10$), GraphIt ($t=1$), \bufGraph{}},
    xtick=data,
    x tick label style={rotate=45,anchor=east},
    ]
\addplot+[ybar, fill=black,draw=black] plot coordinates {(Ligra ($t=10$),39) (Ligra ($t=1$),54) (Gemini ($t=10$),34) (Gemini ($t=1$),38) (GraphIt ($t=10$),40) (GraphIt ($t=1$),55) (\bufGraph{},19)};
\addplot+[ybar, fill=white, draw=black] plot coordinates {(Ligra ($t=10$),61) (Ligra ($t=1$),46) (Gemini ($t=10$),66) (Gemini ($t=1$),62) (GraphIt ($t=10$),60) (GraphIt ($t=1$),45) (\bufGraph{},81)};
%\addplot+[ybar] plot coordinates {(tool1,0) (tool2,0) (tool3,0) (tool4,3) (tool5,1) (tool6,1) (tool7,0)};
\legend{\strut mem unit stalled, \strut mem unit not stalled}
\end{axis}
\end{tikzpicture}
	\captionsetup{aboveskip=2pt}
	\caption{Time breakdown of memory units in the four GPSs, solving NCP on \texttt{Lj}.}
	\label{fig:breakdown}
\end{figure}
%%%%%%%%%%%%%%%%%%%%%%%%%%%%%%%%%
%%  Scalability
%%%%%%%%%%%%%%%%%%%%%%%%%%%%%%%%%
%\noindent\fbox{\parbox{\linewidth}{\textbf{Finding:} \bufGraph{} shows excellent scalability with more threads enabled, and the high throughput of processing growing numbers of FPP queries.}}

\subsection{Scalability}
Figure~\ref{fig:cpu-cont} shows \bufGraph{}'s speedups with the number of threads varied, solving NCP on different graphs. Each thread is assigned to a unique physical core. \SR{\bufGraph{} can achieve $7-8\times$ speedups when scaling up to 10 cores (with hyper-threading disabled) for most of the graphs. However, the speedup of \bufGraph{} on \texttt{Tw} is limited to $5.2\times$. The major reason is the high skewness of \texttt{Tw} graph, which may cause a severe unbalanced workload among threads and thus low CPU utilization. We leave adaptive workload assignments as an enhancement of the system in the future.}

%%%%%%%%%%%%%%%%%%%%%%%%%%%%%%%%%
%%  Scalability
%%%%%%%%%%%%%%%%%%%%%%%%%%%%%%%%%
We analyze the scalability of \bufGraph{} on solving different numbers of FPP queries. We also vary the query types with two new types: DFSs (Depth First Searches) and RWs (Random Walks). For RW, we use the setting from a previous study~\cite{alamgir2010multi}. We measure the throughput to study the scalability of different query types.

Figure~\ref{fig:throughput} plots the normalized throughput of \bufGraph{}. \SR{Usually, the query results will be too large to fit in the main memory when the query count is too large, i.e., 10,000 queries. In this case, we divide queries into ten batches so that each batch of 1,000 queries can fit into the main memory of the machine.} We can observe from the figure that 1) \bufGraph{}'s throughput increases with more FPP queries provided, and 2) \bufGraph{} shows high throughput improvement on queries like PPR, DFS, and RW because these queries have good temporal localities, tending to run many iterations within partitions, 2) although the throughput does not increase much further on SSSPs and BFSs, the processing in \bufGraph{} can remain at the high processing performance, \SR{without obvious penalty caused by redundant operations among more queries}.

\begin{figure}[t]
	\centering
	\begin{minipage}[t]{.22\textwidth}
		\centering
		\begin{tikzpicture}
    \begin{axis}[
        ylabel=Speedup,
        xlabel=Number of available threads,
        width=45mm,
        height=50mm,
        ylabel near ticks,
        xlabel near ticks,
yticklabel style = {font=\normalsize},
xticklabel style = {font=\normalsize},
ylabel style={font=\normalsize},
xlabel style={font=\normalsize},
        symbolic x coords={1,2,3,4,5,6,7,8,9,10},
        line width=0.5,
        xtick=data,
        enlarge x limits=0.05,
        ytick distance = 2,
        ymin = 0,
        ymax = 9,
        legend columns=1,
        legend style={
            at={(-0.04, 1.02)},
			anchor=north west,
			draw=none,
			fill=none,
			font=\scriptsize,
		},
        legend image code/.code={\draw[mark repeat=2,mark phase=2]
			plot coordinates {
				(0cm,0cm)
				(0.15cm,0cm)        %% default is (0.3cm,0cm)
				(0.3cm,0cm)         %% default is (0.6cm,0cm)
			};%
		},
        %extra y ticks = 1,
        %extra y tick labels={1},
        %extra y tick style={grid=major,major grid style={dashed, draw=red}},
        ]
    %\addplot[mark=*, mark options={fill=white}, draw = blue, line width = 1, mark size = 2] plot coordinates {(1, 1)(2,1.082474227)(3,1.148969072)(4,1.197938144)(5,1.239690722)(6,1.277319588)(7,1.317010309)(8,1.336082474)(9,1.377835052)(10,1.402061856)};
    % \addplot table [x=index, y=Ligra, col sep=space] {plot/contention.csv};
    % \addlegendentry{Ligra}
    % \addplot table [x=index, y=Gemini, col sep=space] {plot/contention.csv};
    % \addlegendentry{Gemini}
    % \addplot table [x=index, y=GraphIt, col sep=space] {plot/contention.csv};
    % \addlegendentry{GraphIt}
    \addplot[mark=*,mark size=1pt,mark options={scale=2, color=\ClgC, draw=black}, color=\ClgC] table [x=index, y=Or, col sep=space] {plot/contention.csv};
    \addlegendentry{\texttt{Or}}
    \addplot[mark=square*,mark size=1pt,mark options={scale=2, color=\CgmC, draw=black}, color=\CgmC] table [x=index, y=Wk, col sep=space] {plot/contention.csv};
    \addlegendentry{\texttt{Wk}}
    \addplot[mark=triangle*,mark size=1pt,mark options={scale=2, color=\CgtC, draw=black}, color=\CgtC] table [x=index, y=Lj, col sep=space] {plot/contention.csv};
    \addlegendentry{\texttt{Lj}}
    \addplot[mark=diamond*,mark size=1pt,mark options={scale=2, color=yellow, draw=black}, color=yellow] table [x=index, y=Pt, col sep=space] {plot/contention.csv};
	\addlegendentry{\texttt{Pt}}
    \addplot[mark=pentagon*,mark size=1pt,mark options={scale=2, color=orange, draw=black}, color=orange] table [x=index, y=Tw, col sep=space] {plot/contention.csv};
	\addlegendentry{\texttt{Tw}}
    \end{axis}

\end{tikzpicture}
		\captionsetup{aboveskip=-6pt}
		%\captionsetup{belowskip=-10pt}
		\caption{Speedups with more threads enabled in \bufGraph{}, solving NCP on different graphs.}
		\label{fig:cpu-cont}
	\end{minipage}
	\hfill
	\begin{minipage}[t]{0.22\textwidth}
		\centering
		\begin{tikzpicture}
    \begin{axis}[
yticklabel style = {font=\normalsize},
xticklabel style = {font=\normalsize},
ylabel style={font=\normalsize},
xlabel style={font=\normalsize},
        width=45mm,
        height=48mm,
        line width=0.5,
        %ymode=log,
        %log y ticks with fixed point,
        ylabel=Normalized throughput,
        xlabel=Number of queries,
        ymax = 18,
        %ytick distance = 2000,
        %ybar,bar width=4pt,
        xtick = {1, 2, 3, 4, 5},
        xticklabels = {$10^0$, $10^1$, $10^2$, $10^3$, $10^4$},
        legend columns=1,
        legend style={
            at={(-0.04, 1.02)},
            anchor=north west,
            draw=none,
            fill=none,
            font=\scriptsize,
        },
        legend image code/.code={\draw[mark repeat=2,mark phase=2]
            plot coordinates {
                (0cm,0cm)
                (0.15cm,0cm)        %% default is (0.3cm,0cm)
                (0.3cm,0cm)         %% default is (0.6cm,0cm)
            };%
        },
    ]
    \addplot[mark=*,mark size=1pt,mark options={scale=2, color=\ClgC, draw=black}, color=\ClgC] table [x=index, y=PPR, col sep=space] {plot/throughput.csv};
    %\addlegendentry{PPR on \texttt{Lj}}
    \addlegendentry{PPR on \texttt{Lj}}
    \addplot[mark=square*,mark size=1pt,mark options={scale=2, color=\CgmC, draw=black}, color=\CgmC] table [x=index, y=DFS, col sep=space] {plot/throughput.csv};
    %\addlegendentry{DFS on \texttt{Us}}
    \addlegendentry{DFS on \texttt{Tw}}
    \addplot[mark=triangle*,mark size=1pt,mark options={scale=2, color=\CgtC, draw=black}, color=\CgtC] table [x=index, y=RW, col sep=space] {plot/throughput.csv};
    %\addlegendentry{RW on \texttt{Lj}}
    \addlegendentry{RW on \texttt{Us}}
    \addplot[mark=diamond*,mark size=1pt,mark options={scale=2, color=yellow, draw=black}, color=yellow] table [x=index, y=SSSP, col sep=space] {plot/throughput.csv};
    %\addlegendentry{SSSP on \texttt{Us}}
    \addlegendentry{SSSP on \texttt{Us}}
    \addplot[mark=pentagon*,mark size=1pt,mark options={scale=2, color=orange, draw=black}, color=orange] table [x=index, y=BFS, col sep=space] {plot/throughput.csv};
    %\addlegendentry{BFS on \texttt{Us}}
    \addlegendentry{BFS on \texttt{Tw}}
    \end{axis}
\end{tikzpicture}
		\captionsetup{aboveskip=-6pt}
		%\captionsetup{belowskip=-10pt}
		\caption{Scalability of \bufGraph{} on solving FPP queries at different scales.}
		\label{fig:throughput}
	\end{minipage}
\end{figure}

%%%%%%%%%%%%%%%%%%%%%%%%%%%%%%%%%
%%  Partition & Cache Size
%%%%%%%%%%%%%%%%%%%%%%%%%%%%%%%%%
%\noindent\fbox{\parbox{\linewidth}{\textbf{Finding (4):} The proposed techniques accumulatively improve the performance of \bufGraph{}.}}
\subsection{Effects of Partition and Cache Size} \label{subsec:eval:partition}

\textbf{Effect of partition methods.}
\SR{Partition sizes and methods play important roles in \bufGraph{} setting. There are overwhelming studies in graph partitioning methods, typically defined with two major objectives: load balance and minimum cuts (vertex or edge). As there is only one partition processed at the same time in \bufGraph{}, unbalanced partitioning does not affect the load among threads. The objective of load balance is not necessary here. Minimizing cuts is more important for \bufGraph{}. We use METIS as it is one of the state-of-the-art tools for edge-cut partitioning.}

\SR{First, to evaluate the hypothesis, we evaluate the effects of some common graph partition methods including METIS, random partitioning, Gemini's lightweight partition~\cite{zhu2016gemini}, and GridGraph's 2D partition~\cite{zhu2015gridgraph}. \bufGraph{} on METIS partition shows $14.1\times$, $4.2\times$, and $8.3\times$ speedups over a random partition, Gemini's lightweight partition, and GridGraph's 2D partition, respectively, when executing LL and BC on different graphs. Similarly, \bufGraph{} on METIS partition shows $1.1-3.6\times$ speedups over other partitioning methods when executing NCP on different web and social networks.}

\vspace{1mm}\noindent
\textbf{Effect of partition and cache sizes.} We first evaluate the effects of partition sizes. Figure~\ref{fig:app:part} shows the normalized execution time of \bufGraph{} on different partition sizes, on processing two applications using four graphs. The results show that using LLC-size partitions achieves the best performance for most cases. As the graph partition size increases, there will be fewer partitions to schedule but higher cache thrashing. Further, if we divide the graph into small partitions, there will be too many partitions to schedule and the runtime overhead will be high.

Evaluating the effect of cache by executing \bufGraph{} on other machines may not isolate the impact of cache only, as we do not have CPUs with different LLC sizes but with the same architectures on other aspects. Therefore, we refer to the evaluation in Figure~\ref{fig:app:part}. It helps us to understand the performance on smaller LLC, where graph partitions are smaller. \bufGraph{} with the partition size equals to LLC shows $1.1-4.1\times$ speedups over the settings with partition size of 1/2 and 1/4 of LLC. Thus, we expect that \bufGraph{} could achieve better performance on machines with larger caches.

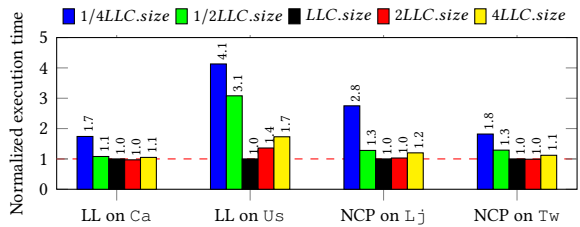
\begin{figure}[t]
	\centering
	\begin{tikzpicture}
    \begin{axis}[
        y tick label style={
            /pgf/number format/fixed,
            /pgf/number format/fixed zerofill,
            /pgf/number format/precision=0,
        },
yticklabel style = {font=\footnotesize},
xticklabel style = {font=\footnotesize},
ylabel style={font=\footnotesize},
xlabel style={font=\footnotesize},
        ytick distance = 1,
        enlarge x limits=0.15,
        width=85mm,
        height=36mm,
        ymin=0,
        ymax=5,
        %ymode=log,
        %log y ticks with fixed point,
        ylabel=Normalized execution time,
        ybar=0pt,
        bar width=6pt,
        xtick = {1, 2, 3, 4},
        xticklabels = {LL on \texttt{Ca}, LL on \texttt{Us}, NCP on \texttt{Lj}, NCP on \texttt{Tw}},
        extra y ticks = 1,
        extra y tick labels={},
        extra y tick style={grid=major,major grid style={dashed, draw=red}},
        legend columns=-1,
        legend style={
            at={(0.5, 1.265)},
            anchor=north,
            draw=none,
            font=\scriptsize,
        },
        nodes near coords style={
	    	font=\tiny,
	    	/pgf/number format/fixed,
	    	/pgf/number format/fixed zerofill,
	    	/pgf/number format/precision=2,
	    	anchor=south,
	    	%inner xsep=0pt,
	    	%shift={(axis direction cs:0,-\rawy)}},
	    	%shift={(axis direction cs:0,0)}
	    	xshift=3.65\pgfkeysvalueof{/pgf/bar width},
	    	yshift=-1.65\pgfkeysvalueof{/pgf/bar width},
    	},
        legend image code/.code={\draw [#1] (0cm,-0.1cm) rectangle (0.15cm,0.1cm);},
                nodes near coords style={
            font=\tiny,
            /pgf/number format/fixed,
            /pgf/number format/fixed zerofill,
            /pgf/number format/precision=1,
            anchor=west,
            rotate=90,
            %inner xsep=0pt,
            %shift={(axis direction cs:0,-\rawy)}},
            shift={(axis direction cs:0,0.2)}
            %xshift=3.65\pgfkeysvalueof{/pgf/bar width},
            %yshift=-1.65\pgfkeysvalueof{/pgf/bar width},
        },
    ]
    \addplot[fill=\ClgC, nodes near coords] table [x=index, y=1, col sep=space] {plot/app_part.csv};
    \addlegendentry{$1/4LLC.size$}
    \addplot[fill=\CgmC, nodes near coords] table [x=index, y=2, col sep=space] {plot/app_part.csv};
    \addlegendentry{$1/2LLC.size$}
    \addplot[fill=black, nodes near coords] table [x=index, y=3, col sep=space] {plot/app_part.csv};
	\addlegendentry{$LLC.size$}
    \addplot[fill=\CgtC, nodes near coords] table [x=index, y=4, col sep=space] {plot/app_part.csv};
    \addlegendentry{$2LLC.size$}
    \addplot[fill=yellow, nodes near coords] table [x=index, y=5, col sep=space] {plot/app_part.csv};
    \addlegendentry{$4LLC.size$}
    \end{axis}
\end{tikzpicture}
	\captionsetup{aboveskip=0pt}
	\captionsetup{belowskip=-3pt}
	\caption{Normalized execution time of \bufGraph{} on different partition sizes.}
	\label{fig:app:part}
\end{figure}

\section{More Related Works}

Prior works propose leveraging graph properties, e.g., vertex degree and the application-visible data access pattern, to improve cache locality. Gorder~\cite{wei2016speedup} is a reordering method that renumbers vertices to maximize the cache-line locality among vertices whose neighbors have consecutive IDs. ReCALL~\cite{lakhotia2017recall} reorders the blocks of nodes in the same cache line to further improves the reordered graph generated by Gorder. Rabbit ordering~\cite{arai2016rabbit} considers the multi-level cache hierarchy and maps the communities in graphs accordingly. Specifically, the smaller and denser communities are mapped to caches that are closer to the processor. Zhang et al.~\cite{zhang2017making} and Faldu et al.~\cite{faldu2020closer} use frequency-based clustering and the skew in vertex degree distribution to improve the utilization of each cache line. The techniques are desirable because they require no modifications to the graph processing, and are orthogonal to \bufGraph{}.

\section{More Discussions}
\textbf{Atomic-free algorithms for FPP queries.} Existing studies present some atomic-free algorithms~\cite{dhulipala2020connectit, nasre2013atomic} to eliminate the synchronization barrier. While an atomic-free algorithm seems a good idea, it is usually algorithm-dependent and more complicated than our design. The atomic-free algorithm usually requires dedicated implementations and none of the existing graph processing systems uses such approaches. Thus, we focus on more general cases in the paper. Moreover, the existing atomic-free techniques trade off extra but cheap computation for mitigating the costly synchronization on GPUs~\cite{nasre2013atomic}. However, such designs are inefficient for multi-core systems as there are many redundant updates. We can expect that it would result in inferior performance.

Second, as a sanity check, we implement the atomic-free SSSP algorithm based on the implementation of the Bellman-Ford algorithm in~\cite{nasre2013atomic} using Ligra. Particularly, multiple CPU threads can simultaneously update the distance of the same node without synchronization. Any lost update will be reconsidered in the next iteration by leveraging the monotonicity of SSSP computation in the topology-driven manner~\cite{pingali2011tao}. We evaluate the implementation and have the observations as follows. 1) The atomic-free computation is a few times slower than the atomic-based one because of redundant updates and computations. 2) When leveraging inter-query parallelism on the atomic-based implementation, the performance drops around $40\%$ compared to the case of leveraging intra-query parallelism, which is mainly caused by cache thrashing. We will add this part to the appendix of a complete version.

\vspace{1mm}\noindent
\textbf{Limitation.} We can observe that \bufGraph{} consistently outperforms other GPSs in the evaluation of solving FPP queries. However, as \bufGraph{} is specially designed for FPP queries, it does not outperform other GPSs on single-query applications like a single BFS. \SR{Besides, \bufGraph{} only targets massive independent queries. It could be an interesting future study to support applications like multi-commodity flow~\cite{awerbuch1994improved}, where queries are dependent on each other. A promising extension is introducing concurrency control to \bufGraph{}, which could bring interesting research connecting relational transaction processing and graph processing.}

\balance
\else
\fi

\end{document}